\documentclass[10pt, doublecolumn, twoside]{IEEEtran}
\usepackage{amsmath,epsfig}
\usepackage{cite}
\usepackage{graphicx,subfigure}
\usepackage{amssymb}
\usepackage{multirow}
\begin{document}

\newtheorem{lem}{Lemma}
\newtheorem{prop}{Proposition}
\newtheorem{cor}{Corollary}
\newtheorem{remark}{Remark}
\newtheorem{defin}{Definition}
\newtheorem{thm}{Theorem}

\newcounter{MYtempeqncnt}

\title{Joint Wireless Information and Energy Transfer with Reduced Feedback in MIMO Interference Channels}

\author{Jaehyun Park,~\IEEEmembership{Member,~IEEE,}
     Bruno Clerckx,~\IEEEmembership{Member,~IEEE,}
\thanks{J. Park is with the Department of Electronic Engineering, Pukyong National University, Republic of Korea. B. Clerckx is with the Department of Electrical and Electronic Engineering, Imperial College London, United Kingdom and also with School
of Electrical Engineering, Korea University (e-mail:jaehyun@pknu.ac.kr, b.clerckx@imperial.ac.uk). B. Clerckx is the corresponding author.} }

\maketitle

\begin{abstract}
To determine the transmission strategy for the joint wireless information and energy transfer (JWIET) in the MIMO interference channel (IFC), the information access point (IAP) and energy access point (EAP) require the channel state information (CSI) of their associated links to both the information-decoding (ID) mobile stations (MSs) and energy-harvesting (EH) MSs (so-called local CSI). In this paper, to reduce the feedback overhead of MSs for the JWIET in two-user MIMO IFC, we propose a Geodesic energy beamforming scheme that requires partial CSI at the EAP. Furthermore, in the two-user MIMO IFC, it is proved that the Geodesic energy beamforming is the optimal non-cooperative strategy under local CSIT assumption. By adding a rank-one constraint on the transmit signal covariance of IAP, we can further reduce the feedback overhead to IAP by exploiting Geodesic information beamforming. Under the rank-one constraint of IAP's transmit signal, we prove that Geodesic information/energy beamforming approach is the optimal non-cooperative strategy for JWIET in the two-user MIMO IFC. We also discuss the extension of the proposed rank-one Geodesic information/energy beamforming strategies to general K-user MIMO IFC. Finally, by analyzing the achievable rate-energy performance statistically under imperfect partial CSIT, we propose an adaptive bit allocation strategy for both EH MS and ID MS.
\end{abstract}


\section{Introduction}
\label{sec:intro}

During the last decade, there has been a lot of interest to transfer energy wirelessly and recently, radio-frequency (RF) radiation has become a viable source for energy harvesting. Furthermore, due to the popularity of sensors, IoT, smart phones and various energy-consuming applications, the battery limitation of wireless devices becomes one of the main practical challenges in modern wireless communication system. Accordingly, the 4th generation (4G) and beyond 4G standards also consider ways to address battery limitations (e.g. device-to-device communications) \cite{3GPPMTC}. In addition, wireless power consortium was established and is working toward the global standardization of wireless charging technology \cite{WPS}.

Because RF signals carry information as well as energy, ``joint wireless information and energy transfer (JWIET)'' has attracted significant attention very recently \cite{Zhang1, Zhang2, KHuang1, Ozel, RRajesh, ANasir, YLuo, Tutuncuoglu2, KHuang2,  ParkBruno, ParkClerckx2}. Most previous works have studied the fundamental performance limits and the optimal transmission strategies of the JWIET under ideal environments (i.e., perfect full channel state information at the transmitter (CSIT))\footnote{Throughout the paper, the full CSI indicates the instantaneous channel matrix, itself. In contrast, the partial CSI indicates the partial information obtained from the full CSI (e.g., the largest singular value/the associated singular vector of the channel matrix or its long-term statistical information). If the (full/partial) CSI is exact (not contaminated by noise or quantization), it is then referred as perfect CSI.}. For example, assuming the perfect knowledge of full CSIT, the downlink of a cellular system with a single base station (BS) and multiple mobile stations (MSs) has been investigated in \cite{KHuang1}, the cooperative relay system in \cite{ANasir}, the broadcasting system in \cite{Zhang1, Zhang2}, and the multi-user SISO OFDM system in \cite{XZhou}. In addition, there have been several studies of JWIET in the interference channel (IFC) \cite{Tutuncuoglu2, KHuang2, ParkBruno, ParkClerckx2}. Because the interference has different impacts on the performances of information decoding (ID) (negative impact) and energy harvesting (EH) (positive impact) at the receivers, the design of suitable transmission strategies for JWIET is a critical issue especially in IFC. Furthermore, the transmission strategy heavily relies on the knowledge of CSIT. For example, to determine the transmission strategy for JWIET in the MIMO IFC, the information access point (IAP) and energy access point (EAP) require the CSI of their associated links to both the ID MSs and EH MSs (i.e. so-called local CSI). However, in a practical system, the acquisition of full CSIT incurs a large system overhead and is more challenging in the MIMO IFC. There exist few papers that address JWIET with partial CSIT (mainly, the long-term correlation) and robust beamforming schemes accounting for the imperfect full CSIT \cite{ZXiang, XChen}. In \cite{ZXiang}, MISO downlink broadcasting channel with three nodes - one BS, one ID MS, and one EH MS - is considered, while in \cite{XChen}, a single user MISO uplink channel is considered.

In this paper, we address how to reduce the feedback overhead in a two-user MIMO IFC, where one IAP and one EAP coexist by sharing the same spectrum resource and serve one ID MS and one EH MS, respectively, in a fully distributed manner. We note that, to the best of the authors' knowledge, it is the first time that the partial CSIT is treated in MIMO IFC accounting for JWIET. Interestingly, we can prove that our proposed non-cooperative strategy with partial CSIT is optimal, contrary to the one currently known in the literature \cite{ParkBruno} that are suboptimal\footnote{Throughout the paper, the notion of optimality is under the assumption that the transmitters are non-cooperative and operate in a distributed manner with local CSIT, unless stated otherwise.}. Because the pseudo-random chaotic waves can be utilized to increase the energy harvesting efficiency \cite{ACollado}, the interference from EAP is assumed not decodable at the ID MS as in \cite{ParkBruno}. Then, the EAP may create a rank-one beam with the aim to either maximize the energy harvested at the EH MS (maximum energy beamforming, MEB) or minimize the interference at the ID MS (minimum leakage beamforming, MLB). In \cite{ParkBruno}, it is proved that to achieve the optimal rate-energy (R-E) performance, the energy transmitter should follow a rank-one beamforming strategy with a proper power control. Accordingly, we first propose a rank-one Geodesic energy beamforming scheme that requires partial CSI at the EAP (mainly, several singular vectors of its associated channel matrices). Here, EAP steers its rank-one beam on the Geodesic curve between MEB and MLB directions. Interestingly, the rationale behind the signal-to-leakage-and-harvested energy-ratio (SLER) beamforming developed in \cite{ParkBruno} can be explained in terms of Geodesic beamforming, but, contrary to the Geodesic beamformer, SLER requires the full CSI of the links to both ID MS and EH MS at the EAP. Furthermore, we prove that the Geodesic energy beamforming scheme is the optimal strategy in the two-user MIMO IFC. Next, by adding a rank-one constraint on the transmit signal covariance of IAP, we can further reduce the feedback overhead to IAP. Here, we propose a Geodesic information beamforming scheme. Under the assumption of the rank-one constraint of IAP's transmit signal, we prove that the Geodesic information/energy beamforming approach is the optimal strategy for JWIET in the two-user MIMO. Motivated by \cite{ParkClerckx2}, the extension of the proposed Geodesic information/energy beamforming strategies to the general K-user MIMO IFC is discussed. Note that to exploit the proposed Geodesic information/energy beamforming, the necessary partial CSI at IAP and EAP is composed of, mainly, the unitary vectors associated with their links to both ID/EH MSs and they can be efficiently quantized using random vector quantization (RVQ) codebooks \cite{N_Jindal2, NRavindran}. Finally, by analyzing the achievable rate-energy performance statistically under the imperfect partial CSIT due to the RVQ, we propose an adaptive bit allocation strategy for both ID/EH MSs that is a function of the path loss and Geodesic angles.

The rest of this paper is organized as follows. In Section
\ref{sec:systemmodel}, we introduce the system model for the two-user
MIMO IFC. In Section \ref{sec:perfectCSIT}, we discuss the transmission
strategies -- MEB, MLB, and SLER -- when full local CSIT is available at both IAP and EAP. In Section \ref{sec:PartialCSIT_EH}, we present the Geodesic energy beamforming when partial CSIT is available at EAP. In Section
\ref{sec:PartialCSIT_IDEH}, when the IAP opts for the rank-one information beamforming, we optimize the information/energy beamforming strategies jointly. In addition, we propose the Geodesic information/energy beamforming schemes and present the extension of the proposed schemes to the general K-user MIMO IFC. In Table \ref{table_intro}, we summarize the available CSIT and the rank $r$ of the transmit signal covariance at IAP and EAP. In Sections \ref{sec:Bitallocation}, we discuss the adaptive bit allocation strategy for both ID/EH MSs. In Section \ref{sec:simulation}, we provide several simulation results and in Section \ref{sec:conc} we give our conclusion.

{\scriptsize{
{\renewcommand\baselinestretch{1.}{
\begin{table}[t]
\renewcommand{\arraystretch}{1.}
\caption{Available CSIT and the rank $r$ of the transmit signal covariance}
\label{table_intro}
\begin{center}
\!\!\begin{tabular}{|c|c|c|c|}\hline
& Section III & Section IV & Section V\\
\hline \multirow{2}*{EAP} & Full CSIT & Partial CSIT  & Partial CSIT\\
 & ($r=1$) & ($r=1$) & ($r=1$)\\
\hline \multirow{2}*{IAP} & Full CSIT  & Full CSIT & Partial CSIT \\
 & ($r\geq 1$) &  ($r\geq 1$) & ($r=1$)\\
\hline\hline K-user & \multirow{2}*{\cite{ParkClerckx2}} & Extendable &  \multirow{2}*{Section V.C} \\
extension & & with \cite{ParkClerckx2} & \\
\hline
\end{tabular}\!
\end{center}
\end{table}
}}
}}

Throughout the paper, matrices and vectors are represented by bold
capital letters and bold lower-case letters, respectively. The
notations $({\bf A})^{H} $, $({\bf A})^{\dagger} $, $({\bf A})_i$,
$[{\bf A}]_i$, $tr({\bf A})$, and $\det({\bf A})$ denote the
conjugate transpose, pseudo-inverse, the $i$th row, the $i$th
column, the trace, and the determinant of a matrix ${\bf A}$,
respectively. The matrix norm $\|{\bf A}\|$ and $\|{\bf A}\|_F$
denote the 2-norm and Frobenius norm of a matrix ${\bf A}$,
respectively, and the vector norm $\|{\bf a}\|$ denotes the 2-norm
of a vector ${\bf a}$. In addition, $(a)^+ \triangleq \max (a, 0)$
and ${\bf A} \succeq {\bf 0}$ (resp. ${\bf A} \succ {\bf 0}$) means that a matrix ${\bf A}$ is positive
semi-definite (resp. definite). Finally, ${\bf I}_{M}$ denotes the $M \times M$
identity matrix and $\lceil\cdot\rfloor $ denotes the rounding operation.

\section{System model}
\label{sec:systemmodel}

\begin{figure}
\begin{center}
\begin{tabular}{c}
\includegraphics[height=3.6cm]{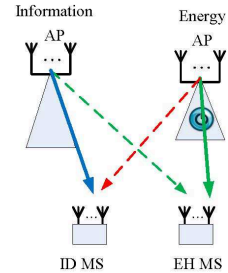}
\end{tabular}
\end{center}
\caption[JWIET_two_userIC_block]
{ \label{JWIET_two_userIC_block} Two-user MIMO IFC in ($EH_1$,
$ID_2$) mode.}
\end{figure}
We consider a two-user MIMO IFC where one IAP with $M_{I,T}$ transmit antennas and one EAP with $M_{E,T}$ transmit antennas, coexist by sharing the same spectrum resource and, respectively, serve one ID MS with $M_{I,R}$ receive antennas and one EH MS with $M_{E,R}$ receive antennas, as shown in Fig. \ref{JWIET_two_userIC_block}.\footnote{Our system model can generalize the scenario where an EAP is additionally deployed in the conventional single-cell (information) downlink system. Furthermore, motivated by \cite{ParkClerckx2}, the extension
to the general K-user MIMO IFC is discussed in Section \ref{ssec:KuserIFC}.} Without loss of generality, EAP (or, energy transmitter) and EH MS are indexed as the first transceiver pair and IAP (or, information transmitter) and ID MS are indexed as the second transceiver pair. In addition, we assume that $M_{I,T}=M_{E,T}= M_{I,R}=M_{E,R}= M$ (i.e., the square matrices) for the sake of readability but the same algorithms can be straightforwardly applied/extended to general matrix sizes. Assuming a frequency flat fading channel, which is static over several frames, the received signal ${\bf y}_i \in \mathbb{C}^{M \times 1}$ for $i=1, 2$ can then be written as
\begin{eqnarray}\label{Sys_1}\nonumber
{\bf y}_1 = {\bf H}_{11} {\bf x}_1 +{\bf H}_{12} {\bf x}_2 + {\bf
n}_1,
\\
{\bf y}_2 = {\bf H}_{21} {\bf x}_1 +{\bf H}_{22} {\bf x}_2 + {\bf
n}_2,
\end{eqnarray}
where ${\bf n}_i \in \mathbb{C}^{M\times 1} $ is a complex white
Gaussian noise vector with a covariance matrix $\sigma_n^2{\bf
I}_{M}$ and ${\bf H}_{ij} \in \mathbb{C}^{M\times M}$ is the
frequency-flat fading channel from the $j$th
transmitter to the $i$th MS whose elements are independent and identically distributed (i.i.d.) zero-mean complex Gaussian random variables (RVs) with a unit variance for $i=j$ and a variance $\alpha_{ij}$ for $i\neq j$. Here, $\alpha_{ij}\in [0,1]$ describes the relative path loss of the cross link compared to the direct link. The vectors ${\bf x}_1, {\bf x}_2 \in \mathbb{C}^{M \times 1}$ are the energy/information transmit
signals, respectively, and they have a transmit power constraint as
\begin{eqnarray}\label{Sys_2}
E[\|{\bf x}_j\|^2] \leq P_{T,j}  {\text{ for }} j=1 \text{ and } 2.
\end{eqnarray}
In this paper, $P_{T,1}=P_{T,2} = P$ for readability purpose, and the SNR is defined as $ SNR = \frac{P}{\sigma_n^2}$. Throughout the paper, to ease readability, it is assumed without loss of generality that
$\sigma_n^2 = 1$, unless otherwise stated. General environments,
characterized by other values of the channel/noise power, can be
described simply by adjusting $P$. Throughout the paper, the singular value decomposition (SVD) of ${\bf H}_{ij}$ can be given as
\begin{eqnarray}\label{Sys_2_rev_Geo}
{\bf H}_{ij} = {\bf U}_{ij}{\bf \Sigma}_{ij}{\bf V}_{ij}^H, \quad {\bf \Sigma}_{ij} = diag\{\sigma_{ij, 1 },..., \sigma_{ij,
M } \},
\end{eqnarray}
where ${\bf U}_{ij}$ and ${\bf V}_{ij}$ are $M \times M$ unitary matrices and $\sigma_{ij, 1}\geq...\geq \sigma_{ij, M }$.

Note that because the pseudo-random chaotic waves can be utilized to increase the energy harvesting efficiency \cite{ACollado}, the interference from EAP is assumed not decodable at the ID MS \cite{ ParkBruno}. The achievable rate at
ID MS, $R_2$, is then given by \cite{Scutari}
\begin{eqnarray}\label{Sys_3_revised}
R_2 = \log \det ({\bf I}_{M} +{\bf H}_{22}^H{\bf R}_{-2}^{-1}{\bf
H}_{22}{\bf Q}_2 ),
\end{eqnarray}
where ${\bf R}_{-2}$ indicates the covariance matrix of noise and
interference at the ID MS, i.e., ${\bf R}_{-2} = {\bf I}_{M} + {\bf H}_{21}{\bf Q}_1 {\bf H}_{21}^H$.
Here, ${\bf Q}_j = E [{\bf x}_j {\bf x}_j^H ]$ denotes the covariance matrix of the transmit signal at the $j$th transmitter and, from (\ref{Sys_2}), $tr({\bf Q}_j) \leq P $. At the EH MS, the total harvested power $E_1$ (more exactly, harvested energy normalized by the baseband symbol period) is given by
\begin{eqnarray}\label{Sys_3}\nonumber
E_1 &=& \zeta_1 E[\|{\bf y}_1 \|^2 ] = \zeta_1 tr\left( \sum_{j=1}^2{\bf H}_{1j}{\bf Q}_j{\bf
H}_{1j}^H +{\bf I}_{M} \right),
\end{eqnarray}
where $\zeta_1$ denotes the efficiency constant for converting the
harvested energy to electrical energy to be stored \cite{Vullers,
Zhang1}. For simplicity, it is assumed that $\zeta_i=1$ and the
noise power is negligible compared to the transferred energy from
either EAP or IAP.\footnote{Even though, throughout the paper, the harvested energy due to the background additive noise at EH receiver for ease of explanation, our analysis can be extended to the scenario of the non-negligible additive noise without difficulty.} That is,
\begin{eqnarray}\label{Sys_4}E_1 &\approx& tr\left( \sum_{j=1}^2{\bf
H}_{1j}{\bf Q}_j{\bf H}_{1j}^H \right)\nonumber \\&=& tr\left( {\bf H}_{11}{\bf Q}_1{\bf H}_{11}^H \right)+tr\left(
{\bf H}_{12}{\bf Q}_2{\bf H}_{12}^H \right)=
E_{11}+E_{12},
\end{eqnarray}
where $E_{ij}=tr\left( {\bf H}_{ij}{\bf Q}_j{\bf H}_{ij}^H
\right)$ denoting the energy transferred from the $j$th
transmitter to the $i$th MS.
Then, the achievable rate-energy region is given as
\begin{eqnarray}\label{oneIDoneEH_1}
\!\!&\!C_{R\!-\!E} (P) \!\triangleq  \!\Biggl\{ \!(R, E) : R \leq
\log \det({\bf I}_{M} + {\bf H}_{22}^H{\bf R}_{-2}^{-1}{\bf
H}_{22}{\bf Q}_2 ),\!&\!\nonumber\\
\!\!&\!E \!\leq \!\sum_{\!j\!=\!1}^{\!2} tr ({\bf H}_{1j} {\bf
Q}_j {\bf H}_{1j}^H), tr({\bf Q}_j)\!\leq \!P, {\bf Q}_j\!\succeq
\!{\bf 0}, j\!=\!1,\!2\! \Biggr\}\!.\!&\!
\end{eqnarray}

\section{Full local CSIT at both information/energy transmitters}\label{sec:perfectCSIT}
In this section, we briefly review the JWIET transmission strategy for two user MIMO IFC \cite{ParkBruno},
assuming that both EAP and IAP have the full knowledge of the CSI of their associated links (i.e. the links between a transmitter and all
MSs) but do not share those CSI between them (i.e. full local CSIT).
In \cite{ParkBruno}, a necessary condition of the optimal transmission strategy has been found for the two-user MIMO IFC with one EH MS and one ID MS, in which the energy transmitter should take a rank-one energy beamforming strategy with a proper power control.
The optimal ${\bf Q}_1$ at the boundary of the achievable rate-energy region has a rank one at most. That is, $rank ({\bf Q}_1) \leq 1$.

Accordingly, the energy transmitter may steer its signal to maximize the energy
transferred to the EH MS and the corresponding transmit
covariance matrix ${\bf Q}_1$ and beamforming vector ${\bf v}_{E}$ are then given by
\begin{eqnarray}\label{MEB}
{\bf Q}_1 = P_1 {\bf v}_{E}{\bf v}_{E}^H,~ {\bf v}_{E} = [{\bf V}_{11}]_1,
\end{eqnarray}
where $0\leq P_1 \leq P$. Here,
the energy harvested from the first transmitter is given by $P_1
\sigma_{11, 1 }^2$. From an ID perspective, the energy transmitter should steer its
signal to minimize the interference power to the ID MS and the corresponding transmit
covariance matrix and beamforming vector ${\bf v}_{L}$ are then given by
\begin{eqnarray}\label{MLB}
{\bf Q}_1 = P_1 {\bf v}_{L}{\bf v}_{L}^H,~{\bf v}_{L} = [{\bf V}_{21}]_M,
\end{eqnarray}
where $0\leq P_1 \leq P$. Then,
the energy harvested from the first transmitter is given by $P_1
\|{\bf H}_{11} [{\bf V}_{21}]_M\|^2$.
Because MEB and MLB strategies are developed according to different aims - either
maximizing transferred energy to EH MS or minimizing interference
(or, leakage) to ID MS, respectively, they have their own weakness - causing either large interference to ID MS or insufficient energy to be harvested at EH MS. To maximize the transferred energy to EH MS and
simultaneously minimize the leakage to ID MS, we have also introduced the metric signal-to-leakage-and-harvested energy ratio
(SLER) as
\begin{eqnarray}\label{GSVD1}
\!\!SLER \!&\!\!\!=\!\!\!&\!\frac{\|{\bf H}_{11}{\bf v} \|^2}{\!\|\!{\bf H}_{21}\!{\bf v} \!\|^2\!
+\! max(\bar E \!-\!P\!\|\!{\bf H}_{11}\!\|^2\! ,0)\!} \nonumber \\\!\!&\!\!\!=\!\!&\!\frac{{\bf v}^H{\bf H}_{11}^H{\bf H}_{11}{\bf v}_k }{\!{\bf
v}_k^H\!\left(\!{\bf H}_{21}^H\!{\bf H}_{21} +{max\!({\bar E}/{P} \!-\!\|{\bf
H}_{11}\|^2\! ,0\!)\!}{\bf I}_M \!\right)\!{\bf v}_k \!}\!,\!
\end{eqnarray}
which balances both metrics -
energy maximization to EH MS and leakage minimization to ID MS, as confirmed in \cite{ParkBruno}.
The corresponding transmit covariance matrix and beamforming vector ${\bf v}_{S}$ that maximizes SLER of
(\ref{GSVD1}) are then given by
\begin{eqnarray}\label{GSVD3}
{\bf Q}_1 = P_1 {\bf v}_{S} {\bf v}_{S}^H,~{\bf v}_{S} = \frac{ \bar{\bf v}}{\|\bar{\bf v}\|},
\end{eqnarray}
where $\bar{\bf v}$ is the generalized eigenvector associated with
the largest generalized eigenvalue of the matrix pair $({\bf
H}_{11}^H{\bf H}_{11}, {\bf H}_{21}^H{\bf H}_{21} +{\max(\bar E
/{ P} -\|{\bf H}_{11}\|^2 ,0)}{\bf I}_M)$.
Here, $\bar{\bf v}$
can be efficiently computed by using a GSVD algorithm \cite{JPark}.

\subsection{Optimization for the achievable Rate-Energy region}\label{sec:REregion_perfectCSIT}
Given that ${\bf Q}_1$ is chosen among (\ref{MEB}), (\ref{MLB}), and (\ref{GSVD3}), the
achievable rate-energy region is then given as:
\begin{eqnarray}\label{oneIDoneEHmax1}
\!\!&\!C_{R-E} (P) =  \Biggl\{ (R, E) : R = R_2, E =E_{11} + E_{12}, \quad \!&\!\nonumber\\
\!\!&\!\!R_2 \!\leq \!\log \det({\bf I}_{M} \!+\! {\bf H}_{22}^H{\bf
R}_{-2}^{-1}{\bf H}_{22}{\bf Q}_2 ), E_{12} \!\leq \! tr ({\bf
H}_{12} {\bf Q}_2 {\bf H}_{12}^H)\!, \!\!&\!\!\nonumber\\\!\!\!&\!\!
tr({\bf Q}_2)\leq P, {\bf Q}_2\succeq {\bf 0}, 0 \leq P_1 \leq P
\Biggr\},\!\!\!&\!\!\!
\end{eqnarray}
where
\begin{eqnarray}\label{REregion_2}
E_{11} = \omega_1 P_1, {\bf R}_{-2} = {\bf I}_{M}+ P_1{\bf \Omega}_{21},
\end{eqnarray}
with $(\omega_1,  {\bf \Omega}_{21})= (\|{\bf H}_{11} {\bf v}_p \|^2,  {\bf H}_{21} {\bf v}_p {\bf v}_p^H{\bf H}_{21}^H )$ and $p \in\{E, L, S\}$ for MEB, MLB, and SLER maximization beamforming, respectively.

Accordingly, by letting $\tilde{\bf H}_{22} = {\bf
R}_{-2}^{-1/2}{\bf H}_{22}$, we have the following optimization
problem for the rate-energy region of
(\ref{oneIDoneEHmax1})
\begin{eqnarray}\label{oneIDoneEHmax3}
\!\!(\!P1\!)\! \underset{P_1, {\bf Q}_2}{\text{ maximize}}& J \triangleq \log \det
({\bf I}_{M} + \tilde{\bf H}_{22}{\bf Q}_2\tilde{\bf H}_{22}^H
)\\\label{oneIDoneEHmax3_1} \!\!{\text{subject to}}\!&\!\!tr({\bf
H}_{12}{\bf Q}_2{\bf H}_{12}^H) \geq \max(\bar E \!-\!
E_{11},0)\!\\\label{oneIDoneEHmax3_2} &tr({\bf Q}_2) \leq P,~{\bf
Q}_2 \succeq{\bf 0},~0 \leq P_1 \leq P,
\end{eqnarray}
where $\bar E$ can take any value less than $E_{\max}$. $E_{\max}$ denotes
the maximum energy transferred from both transmitters, i.e., $E_{\max} = \omega_1P_1 + P\sigma_{12,1}^2$ where $\sigma_{12,1}$ denotes the largest singular value of ${\bf H}_{12}$. Note that because $E_{11}$ in (\ref{oneIDoneEHmax3_1}) and
$\tilde{\bf H}_{22}$ in (\ref{oneIDoneEHmax3}) depend on
$P_1(\leq P)$, we identify the achievable R-E region iteratively
as:

\vspace*{2pt}Algo. 1. {\it{\underline{Iterative algorithm for the
achievable R-E region:}}}
\begin{enumerate}
\item Initialize $n=0$, $P_1^{(0)}=P$,
\begin{eqnarray}\label{ReviseAlgo2_1}
E_{11}^{(0)} = \omega_1 P_1^{(0)}, {\bf R}_{-2}^{(0)} = {\bf I}_{M}+ P_1^{(0)}{\bf \Omega}_{21}.
\end{eqnarray}
\item For $n=0:N_{max}$
\begin{enumerate}
\item Solve the optimization problem (P1) for ${\bf Q}_2^{(n)}$ as
a function of $E_{11}^{(n)}$ and ${\bf R}_{-2}^{(n)}$.
\item If $tr ({\bf H}_{12} {\bf Q}_2^{(n)} {\bf H}_{12}^H) +
E_{11}^{(n)}
>\bar E$
\begin{eqnarray}\label{eqnRevise_Algo_2_2} P_1^{(n+1)} =
max\left(P_1^{(n)} - \Delta, 0\right),
\end{eqnarray}
where the step size $\Delta$ is given by a value on $[0, \Delta_{max}]$ with $\Delta_{max}=\frac{ tr ({\bf H}_{12} {\bf Q}_2^{(n)}{\bf
H}_{12}^H) + \omega_1 P_1^{(n)} - \bar E }{\omega_1}$.
\item Else if, $tr ({\bf H}_{12} {\bf Q}_2^{(n)} {\bf H}_{12}^H) +
E_{11}^{(n)}
=\bar E$, then, $P_1^{(n+1)} =
\gamma_1 P_1^{(n)},$ where $\gamma_1 ( <1)$ is a power reduction factor.
\item Update
$E_{11}^{(n+1)}$ and ${\bf R}_{-2}^{(n+1)}$ with $P_1^{(n+1)}$
similarly to (\ref{ReviseAlgo2_1}).
\end{enumerate}
\item Finally, the boundary point of the achievable R-E region is given as $(R, E) =(\log \det
({\bf I}_{M} + \tilde{\bf H}_{22}{\bf Q}_2^{(N_{max}+1)}\tilde{\bf
H}_{22}^H ),~ E_{11}^{(N_{max}+1)} + tr ({\bf H}_{12} {\bf
Q}_2^{(N_{max}+1)} {\bf H}_{12}^H)) $.
\end{enumerate}
 \vspace*{2pt}
In Step 2 of Algorithm 1, the optimization problem (P1) with $E_{11}^{(n)}$ and ${\bf
R}_{-2}^{(n)}$ can be tackled with two different approaches
according to the value of $\bar E$, i.e., $0 \leq \bar E \leq
E_{11}$ and $E_{11} < \bar E \leq E_{\max}$, where we have
dropped the superscript of the iteration index $(n)$ for notation
simplicity. For $0 \leq \bar E \leq E_{11}$, (P1) becomes the
conventional rate maximization problem for single-user effective
MIMO channel (i.e., $\tilde{\bf H}_{22}$) \cite{ParkBruno} resulting in the maximum achievable rate for the given rank-one
strategy ${\bf Q}_1$. For $E_{11} < \bar E \leq E_{\max}$, the optimization problem (P1)
can be solved by a ``water-filling-like'' approach similar to the
one appeared in the joint wireless information and energy
transmission optimization with a single transmitter \cite{Zhang1}.
That is, by defining the Lagrangian function of ($P1$) can be written as
\begin{eqnarray}\label{oneIDoneEHmax8}\nonumber
\!&L({\bf Q}_2, \lambda, \mu) = \log \det ({\bf I}_{M} +
\tilde{\bf H}_{22}{\bf Q}_2\tilde{\bf H}_{22}^H
)  \!&\!\nonumber\\\!&\!+\lambda( tr({\bf H}_{12}{\bf Q}_2{\bf H}_{12}^H) -
(\bar E \!-\! E_1)  ) - \mu (tr({\bf Q}_2) - P),&\!\nonumber
\end{eqnarray}
and the corresponding dual function as $g(\lambda, \mu) = \underset{{\bf Q}_2 \succeq{\bf 0}}{\max} L({\bf
Q}_2, \lambda, \mu)$, the optimal solution is computed from \cite{SBoyd, Zhang1}
\begin{eqnarray}\label{oneIDoneEHmax6}
{\bf Q}_2 &=& {\bf A}^{-1/2}\tilde{\bf V}'_{22}\tilde{\boldsymbol
\Lambda}'\tilde{\bf V}_{22}'^H{\bf A}^{-1/2},
\end{eqnarray}
where $\tilde{\bf V}'_{22}$ is obtained from the SVD of the matrix
$\tilde{\bf H}_{22}{\bf A}^{-1/2}$, i.e., $\tilde{\bf H}_{22}{\bf
A}^{-1/2} =\tilde{\bf U}'_{22}\tilde{\boldsymbol \Sigma}'_{22}
\tilde{\bf V}_{22}'^H$. Here, $\tilde{\boldsymbol \Sigma}'_{22} =
diag\{ \tilde\sigma_{22,1}',...,\tilde\sigma_{22,M}' \}$ with $
\tilde\sigma_{22,1}' \geq...\geq \tilde\sigma_{22,M}' \geq 0$ and
$\tilde{\boldsymbol \Lambda}' = diag\{ \tilde p_1,...,\tilde p_M
\}$ with $\tilde p_i = (1-1/\tilde\sigma_{22,i}'^2)^+ $,
$i=1,...,M$. The parameters $\mu$ and $\lambda$ minimizing
$g(\lambda, \mu)$ can be solved by the subgradient-based
method \cite{Zhang1, XZhao}, where the subgradient of
$g(\lambda, \mu)$ is given by $(tr({\bf H}_{12}{\bf Q}_2{\bf
H}_{12}^H) - (\bar E \!-\! E_1),  P - tr({\bf Q}_2)) $. Because (\ref{oneIDoneEHmax3}) is concave over ${\bf Q}_2$ and monotonically decreasing with respect to $P_1$, we can easily find that every superlevel set $\{{\bf Q}_2, P_1| J({\bf Q}_2, P_1) \geq \alpha\}$ for $\alpha \in \mathbb{R}$ is convex. That is, (\ref{oneIDoneEHmax3}) is quasi-concave \cite{SBoyd} and, because Algorithm 1 converges
monotonically, the converged solution of Algorithm 1 is globally optimal under the local CSIT with a fixed energy beamforming strategy \cite{Bazaraa_book}. See also \cite{ParkBruno} for the details. If we set the maximum power $P_{T,1}$ as 0, Algorithm 1 for the MIMO IFC boils down to that for the MIMO BC in \cite{Zhang1}.
\begin{remark}\label{remark0}
Note that the iterative Algorithm 1 for the optimization of the covariance matrices requires full local CSIT at both energy/information transmitters.
That is, at the energy transmitter, the channel matrices of ${\bf H}_{11}$ and ${\bf H}_{21}$ are required in the computation of $\omega_1$ and, at the information transmitter, the channel matrices of ${\bf H}_{12}$ and ${\bf H}_{22}$ and the interference covariance matrix ${\bf R}_{-2}^{(n)}$ are required in the optimization of ${\bf Q}_2^{(n)}$. Here, ${\bf R}_{-2}^{(n)}$ can be estimated in ID MS and reported to the IAP. In the same manner, the $E_{11}^{(n)}$ needs to be measured and reported for $P_1$ to be adjusted. Note that the feedback overhead of several scalar values such as the target harvested energy $\bar E$, $E_{11}^{(n)}$, and $tr ({\bf H}_{12} {\bf Q}_2^{(n)} {\bf H}_{12}^H)$ are negligible compared to that of the channel matrices. In the next section, motivated by the fact that the SLER maximization beamforming creates a rank-one unit-norm beam with a direction softly bridging MEB and MLB, we develop a Geodesic geometry based beamforming, which reduces the feedback overhead to the energy transmitter.
\end{remark}

\section{Partial CSIT at energy transmitter: Geodesic geometry based feedback reduction}\label{sec:PartialCSIT_EH}

\subsection{Preliminary: Geodesic geometry}\label{ssec:GeodesicPrelim}
Given two points on a manifold, a geodesic is the shortest curve on the manifold between two points. For example, for two points on $M$-dimensional Euclidean space, the geodesic is a line connecting the two points. In contrast, for two points on $M$-dimensional unit-norm Euclidean space, the geodesic is the curve connecting the two points on the $M$-dimensional {\it{unit-norm}} sphere. Then, for any two vectors, ${\bf v}_1$ and ${\bf v}_2$ in $\{ {\bf v} | \|{\bf v}\|^2 = 1, {\bf v} \in \mathbb{C}^{M\times 1} \}$, the vector between them can be computed by using geodesic geometry as \cite{TPande}
\begin{eqnarray}\label{Partial1}
{\bf v}_g(\theta) = {\bf v}_1 {u}_1 \cos(\theta) - ({\bf v}_1)^{\perp}{\bf u}_2 \sin(\theta),
\end{eqnarray}
where ${u}_1$ is the phase difference between ${\bf v}_1$ and ${\bf v}_2$, obtained from
\begin{eqnarray}\label{Partial2}
{\bf v}_1^H {\bf v}_2= {u}_1 \cos \phi_1,
\end{eqnarray}
where $\phi_1$ is the principal angle between ${\bf v}_1$ and ${\bf v}_2$ given as $\phi_1 = \cos^{-1}\left( |{\bf v}_1^H {\bf v}_2| \right) $. Note that this principal angle is the Geodesic distance between ${\bf v}_1$ and ${\bf v}_2$. Here, $({\bf v}_1)^{\perp}\in \mathbb{C}^{M\times (M-1)}$ is the orthogonal completion of ${\bf v}_1$, i.e., $({\bf v}_1)^{\perp}$ spans the column null space of ${\bf v}_1^T$. In addition, $0\leq \theta \leq \phi_1$ and ${\bf u}_2\in \mathbb{C}^{(M-1)\times 1}$ is a unit norm vector such that $\|{\bf u}_2\|^2 = 1$.
Because ${\bf v}_g(\phi_1) = {\bf v}_2$, $({\bf v}_1)^{\perp}{\bf u}_2$ can be given as
\begin{eqnarray}\label{Partial3}
({\bf v}_1)^{\perp}{\bf u}_2 = [{\bf v}_1 {u}_1 \cos(\phi_1) - {\bf v}_2](\sin\phi_1)^{-1}.
\end{eqnarray}

\subsection{Geodesic geometry based rank-one energy beamforming}\label{ssec:GeodesicBF}

Because ${\bf v}_p$, $p\in \{E, L, S\}$ are on the $M$ dimensional unit sphere and the SLER maximization beamforming create a rank-one beam with a direction softly bridging MEB and MLB, from (\ref{Partial1}), we can generate the Geodesic beamforming vector as
\begin{eqnarray}\label{Partial4}
{\bf v}_G(\theta_1) = {\bf v}_E {u}_E \cos(\theta_1) - ({\bf v}_E)^{\perp}{\bf u}_L \sin(\theta_1),
\end{eqnarray}
where $({\bf v}_E)^{\perp}{\bf u}_L = [{\bf v}_E {u}_E \cos(\phi_E) - {\bf v}_L](\sin\phi_E)^{-1}$ and $ {u}_E = \frac{{\bf v}_E^H {\bf v}_L}{ |{\bf v}_E^H {\bf v}_L| }$.
Here, $\phi_E$ is the principal angle between ${\bf v}_E$ and ${\bf v}_L$ given as $\phi_E = \cos^{-1}\left( |{\bf v}_E^H {\bf v}_L| \right) $ and $0\leq \theta_1 \leq \phi_E$. Then, ${\bf v}_G(\theta_1)$ can be rewritten as
\begin{eqnarray}\label{Partial5}
\!\!{\bf v}_G(\theta_1)\! &\!\!=\!\!&\! {\bf v}_E {u}_E \cos(\theta_1) \nonumber\\\!\!&\!\!\!\!&\!- [{\bf v}_E {u}_E \cos(\phi_E) \!- \!{\bf v}_L](\sin\phi_E)^{-1} \sin(\theta_1).
\end{eqnarray}
Interestingly, when $\theta_1$ goes to 0 (resp, $\phi_E$), ${\bf v}_G(\theta_1)$ becomes close to MEB vector (resp, MLB vector) and we can have the following propositions, which are useful to show the optimality of Geodesic beamforming in Proposition \ref{prop2} and Theorems 1 and 2. Their proofs are given in Appendix \ref{appndix1}.
\begin{prop}\label{prop3} The function $E_{11}(\theta_1)\triangleq tr\left({\bf H}_{11}{\bf v}_G(\theta_1){\bf v}_G^H(\theta_1){\bf H}_{11}^H\right)$ is monotonically decreasing with respect to ${\theta}_1$ for the Geodesic energy beamforming.
\end{prop}
\begin{prop}\label{prop3_IN} The function $IN_{21}(\theta_1)\triangleq tr\left({\bf H}_{21}{\bf v}_G(\theta_1){\bf v}_G^H(\theta_1){\bf H}_{21}^H\right)$ is monotonically decreasing with respect to ${\theta}_1$ for the Geodesic energy beamforming.
\end{prop}

To evaluate the achievable region, we jointly optimize $P_1$, $\theta_1$ and ${\bf Q}_2$. The following lemma and proposition are useful in finding the optimal $\theta$ at the boundary points of the achievable R-E region.

\begin{lem}\label{lem1} For a positive semi-definite matrices ${\bf X}$ and ${\bf S}$ (${\bf S}\neq {\bf 0}$), let
\begin{eqnarray}\label{Prop1_1}
f({\bf X}) \triangleq \log \det({\bf I}_{M} + {\bf S}({\bf I}_M + {\bf X})^{-1} ).
\end{eqnarray}
Then, the maximization of $f({\bf X})$ with respect to ${\bf X}$ is equivalent with the minimization of $\det({\bf I}_{M} +{\bf X})$ with respect to ${\bf X}$.
\end{lem}
\begin{proof}
The proof is straightforward; thus it is omitted.
\end{proof}

\begin{prop}\label{prop2} The optimal ${\theta}_1^o$ yielding the boundary point of the achievable $C_{R-E}$ for the Geodesic energy beamforming is given by
\begin{eqnarray}\label{prop2_1}
\theta_1^o =  \underset{0 \leq \theta_1 \leq \phi_0}{\arg} \max \eta(\theta_1) \triangleq \frac{\|{\bf H}_{11}{\bf v}_G(\theta_1)\|^2 }{\|{\bf H}_{21}{\bf v}_G(\theta_1)\|^2},
\end{eqnarray}
where $\phi_0$ is the largest angle satisfying $ P\|{\bf H}_{11}{\bf v}_G(\phi_0)\|^2 = \bar E - E_{12}$. Note that the transferred energy from IAP, $E_{12}$, is upper bounded as $E_{12} \leq P \sigma_{12,1}^2$. If $\theta_1$ is larger than $\phi_0$, resulting in small $E_{11}$, there exists no feasible solution of (P1) to satisfy the constraint (\ref{oneIDoneEHmax3_1}).
\end{prop}
\begin{proof}
See Appendix \ref{appndix3}.
\end{proof}
Note that the range of $\theta_1$ (specifically, $\phi_0$) depends on $E_{12}$. Accordingly, $\theta_1$ can be jointly optimized together with $P_1$ and ${\bf Q}_2$ by modifying Algorithm 1 as:

\vspace*{2pt}Algo. 2. {\it{\underline{Iterative algorithm for the
achievable R-E region}}}\\{\it{\underline{for Geodesic energy beamforming:}}}
\begin{enumerate}
\item Initialize $n=0$ and determine $\phi_0^{(0)}$ such as $\underset{\phi_0}{\arg}\min\left|P tr\left({\bf H}_{11}{\bf v}_G(\phi_0){\bf v}_G^H(\phi_0){\bf H}_{11}^H\right)  -\bar E \right|$.
\item For $n=0:N_{max}$
\!\!\!\begin{enumerate}
\item Find $\theta_1^{(n)}\in [0, \phi_0^{(n)}]$ as (\ref{prop2_1}) and update ${\bf Q}_1^{(n)} = P {\bf v}_G(\theta_1^{(n)}){\bf v}_G^H(\theta_1^{(n)})$ and
\begin{eqnarray}\label{ReviseAlgoGeo2_1}
E_{11}^{(n)} = tr({\bf H}_{11}{\bf Q}_1^{(n)}{\bf H}_{11}^H),\nonumber\\ {\bf R}_{-2}^{(n)} = {\bf I}_{M}+ {\bf H}_{21}{\bf Q}_1^{(n)}{\bf H}_{21}^H,
\end{eqnarray}
\item Find ${\bf Q}_2^{(n)}$ and $P_1^{(n)}$ by using Algorithm 1 with (\ref{ReviseAlgoGeo2_1}).
\item Then, update $\phi_0^{(n)}$ such that
\begin{eqnarray}\label{ReviseAlgoGeo2_4}
P tr\left({\bf H}_{11}{\bf v}_G(\phi_0^{(n)}){\bf v}_G^H(\phi_0^{(n)}){\bf H}_{11}^H\right) =\nonumber\\ \bar E - tr({\bf H}_{12}{\bf Q}_2^{(n)}{\bf H}_{12}^H).
\end{eqnarray}
\end{enumerate}
\item Finally, with $E_{11}^{(N_{max})} = P_1^{(N_{max})} \|{\bf H}_{11}{\bf v}_G(\theta_1^{(N_{max})})\|^2$, the boundary point of the achievable R-E region is given as
\begin{eqnarray}\label{ReviseAlgoGeo2_5}
(R, E) =(\log \det
\left({\bf I}_{M} + \tilde{\bf H}_{22}{\bf Q}_2^{(N_{max})}\tilde{\bf
H}_{22}^H \right),\nonumber\\ E_{11}^{(N_{max})} + tr \left({\bf H}_{12} {\bf
Q}_2^{(N_{max})} {\bf H}_{12}^H\right)).
\end{eqnarray}
\end{enumerate}
 \vspace*{2pt}
\begin{remark}\label{remark1}
From Proposition \ref{prop3}, $E_{11}$ is monotonically decreasing with respect to ${\theta}_1$ and accordingly, in Step 2.a) of Algorithm 2, if $\theta_1$ is larger than $\phi_0^{(n)}$, resulting in small $E_{11}$, there exists no feasible solution of (P1) to satisfy the constraint (\ref{oneIDoneEHmax3_1}). Thanks to Proposition \ref{prop3}, in Step 1 and Step 2.c of Algorithm 2, $\phi_0^{(n)}$ can be efficiently found by using the bisection method \cite{SBoyd}.
\end{remark}
\begin{remark}\label{remark2}
Note that the maximization of $\eta (\theta_1)$ in (\ref{prop2_1}) is analogous to the SLER beamforming. For example, when the required energy $\bar E$ at the EH MS is large, the upper bound of $\theta_1$ in (\ref{prop2_1}) decreases, which implies that the geodesic beamforming becomes close to MEB vector. This observation can also be found in the SLER beamforming. That is, when the required harvested energy is large, the matrix ${\bf H}_{21}^H\!{\bf H}_{21} +{max\!({\bar E}/{P} \!-\!\|{\bf
H}_{11}\|^2\! ,0\!)\!}{\bf I}_M$ in the denominator
of (\ref{GSVD1}) approaches an identity matrix multiplied by a
scalar and the SLER maximizing beamforming is equivalent
with the MEB in (\ref{MEB}). However, while the SLER beamforming requires the full CSIT of the direct/cross links at the energy transmitter, the geodesic beamforming requires only two unit-norm vectors of ${\bf v}_E$ and ${\bf v}_L$, which can be efficiently quantized using a codebook relying on random vector quantization \cite{N_Jindal2} or Grassmannian line packing \cite{Love3}. Note that while Algorithm 1 optimizes $P_1$ and ${\bf Q}_2$ for a fixed energy beamformer, Algorithm 2 can optimize the energy beamforming as well in a distributed manner based on local CSIT.
\end{remark}
\begin{remark}\label{remark3}
Together with (\ref{Prop1_9}) in Appendix A, $ \eta(\theta_1)$ in (\ref{prop2_1}) can be rewritten as
\begin{eqnarray}\label{Partial6}
\!\!\!\eta(\theta_1) \!=\! \!\frac{ \cos^2(\theta_1)\sigma_{11,1}^2 +\sin^2(\theta_1) \| {\bf H}_{11}({\bf v}_E)^{\perp}{\bf u}_L\|^2 }{ \cos^2(\!\phi_E \!-\! \theta_1\!)\sigma_{21,M}^2 \!+\!\sin^2(\!\phi_E \!-\! \theta_1\!) \| {\bf H}_{21}\!({\bf v}_L)^{\!\perp}\!{\bf u}_E\|^2   \!}\!.\!\!\nonumber\\\!\!\!\!\!\!\!\!\!\!
\end{eqnarray}
Accordingly, to find the optimal $\theta_1^o$, the energy transmitter needs to know four additional scalar values of $\sigma_{11,1}^2 $, $\sigma_{21,M}^2$, $ \| {\bf H}_{11}({\bf v}_E)^{\perp}{\bf u}_L\|^2$, and $ \| {\bf H}_{21}({\bf v}_L)^{\perp}{\bf u}_E\|^2 $. The last two of them can be evaluated at each MS from two different reference signals (with $({\bf v}_E)^{\perp}{\bf u}_L$ and $({\bf v}_L)^{\perp}{\bf u}_E$, respectively) of the energy transmitter and reported back to the energy transmitter. Note that, similarly to the way of estimating the numerator of (\ref{Partial6}), in (\ref{ReviseAlgoGeo2_1}) (resp. (\ref{ReviseAlgoGeo2_4})), $E_{11}^{(n)}$ (resp. $tr\left({\bf H}_{11}{\bf v}_G(\phi_0^{(n)}){\bf v}_G^H(\phi_0^{(n)}){\bf H}_{11}^H\right)$) can be evaluated at the energy transmitter with  $\sigma_{11,1}^2 $ and $ \| {\bf H}_{11}({\bf v}_E)^{\perp}{\bf u}_L\|^2$, while the information of ${\bf R}_{-2}^{(n)}$ in (\ref{ReviseAlgoGeo2_1}) is not required at the energy transmitter.
\end{remark}
Note that, the information transmitter still requires CSIT of its direct/cross links to solve (P1) for ${\bf Q}_2^{(n)}$ for given $E_{11}^{(n)}$ and ${\bf R}_{-2}^{(n)}$ in Step 2 of Algorithm 1. In Section \ref{sec:PartialCSIT_IDEH}, to further reduce the feedback overhead to the information transmitter as well, we also propose the geodesic information beamforming by introducing an additional rank one constraint on the information transmitter.

\subsection{Optimality of Geodesic energy beamforming for Rate-Energy region of two-user MIMO IFC}\label{ssec:optimality1}
Motivated by Proposition \ref{prop2}, together with Proposition 2 (rank-one optimality) in \cite{ParkBruno} or Corollary 1 in \cite{ParkClerckx2}, we can derive the following theorem that gives us very important insights into the beamforming strategy that yields the optimal boundary of the achievable rate-energy region in (\ref{oneIDoneEH_1}). Note that for two-user IFC, while \cite{ParkBruno} is only focused on low/high SNR, \cite{ParkClerckx2} addresses the rank-1 optimality for any SNR region.

\begin{thm}\label{thm1} For two-user MIMO IFC (one energy transceiver and one information transceiver), the optimal energy beamforming vector that yields the optimal boundary of the achievable rate-energy region in (\ref{oneIDoneEH_1}) lies in the Geodesic curve between $[{\bf V}_{11}]_1(\triangleq{\bf v}_E )$ and $[{\bf V}_{21}]_M(\triangleq{\bf v}_L )$.
\end{thm}
\begin{proof}
See Appendix \ref{appndix4}.
\end{proof}
\begin{figure}
\begin{center}
\begin{tabular}{c}
\includegraphics[height=3.6cm]{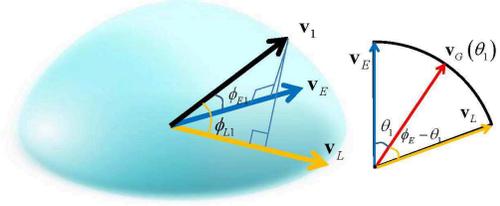}
\end{tabular}
\end{center}
\caption[geodesicOptimal]
{ \label{geodesicOptimal} Geometry of Geodesic beamforming vector.}
\end{figure}
From Theorem \ref{thm1}, because the energy transmitter pursuits two different objectives - maximize the harvesting energy at EH MS and minimizing the interference power to ID MS, if it should have a rank-one beamforming, then the optimal beamforming that yields the optimal boundary of the achievable rate-energy region in (\ref{oneIDoneEH_1}) becomes the Geodesic beamforming.
Interestingly, the optimal energy beamforming vector in (\ref{Partial5}) is a linear combination of MEB (signal maximization) and MLB (interference minimization) vectors and it is reminiscent of the optimal beamforming under the local CSIT in conventional IFC which is a linear combination of a matched filter beamformer (signal maximization) and a zero-forcing beamformer (interference minimization) \cite{LarssonJorswieck}.

\section{Partial local CSIT at both information/energy transmitters}\label{sec:PartialCSIT_IDEH}
Before proposing the geodesic information beamforming, we first present how to optimize the covariance matrices of energy/information transmitters when they both use the rank-one beamforming strategies.

\subsection{Optimization for the achievable Rate-Energy region}
Because both energy/information transmitters have a rank-one beamforming, the achievable rate-energy region is given as:
\begin{eqnarray}\label{PartialCSIT1}
\!&\!C_{R-E} (P) =  \Biggl\{ (R, E) : R = R_2, E =E_{11} + E_{12}, \quad \!&\!\nonumber\\
\!&\!R_2 \!\leq \!\log (1 + P{\bf w}_2^H{\bf H}_{22}^H{\bf
R}_{-2}^{-1}{\bf H}_{22}{\bf w}_2 ), \nonumber\!&\!\\\!&\!E_{12} \!\leq \! P{\bf w}_2^H{\bf H}_{12}^H{\bf
H}_{12} {\bf w}_2, \|{\bf w}_2\|^2 = 1, 0 \leq P_1 \leq P \Biggr\},\!\!\!&\!\!\!
\end{eqnarray}
where $E_{11} = \omega_1 P_1, {\bf R}_{-2} = {\bf I}_{M}+ P_1{\bf \Omega}_{21}$ with $(\omega_1,  {\bf \Omega}_{21})= (\|{\bf H}_{11} {\bf v}_G(\theta_1) \|^2,  {\bf H}_{21} {\bf v}_G(\theta_1) {\bf v}_G^H(\theta_1){\bf H}_{21}^H )$. Here the SLER beamforming is not considered, because the SLER beamforming requires the full CSIT at the energy transmitter. In addition, MLB and MEB can be regarded as a special case of the Geodesic beamforming with $\theta_1 = \{0, \phi_E\}$.

To evaluate the achievable region, we optimize $P_1$, $\theta_1$ and ${\bf w}_2$ under the distributed optimization framework. However, Proposition \ref{prop2} is still valid when the information transmitters have a rank-one beamforming. Therefore, $\theta_1$ can be determined such as (\ref{prop2_1}). Accordingly, we have the following optimization problem for the rate-energy region of
(\ref{PartialCSIT1})
\begin{eqnarray}\label{PartialCSIT2_1}
\!\!(\!P2\!)\! \underset{P_1, {\bf w}_2}{\text{ maximize}}& \log (1 + P{\bf w}_2^H{\bf H}_{22}^H{\bf
R}_{-2}^{-1}{\bf H}_{22}{\bf w}_2 )
\\\label{PartialCSIT2_2} \!\!{\text{subject to}}\!&\!\!P{\bf w}_2^H{\bf H}_{12}^H{\bf
H}_{12} {\bf w}_2 +
E_{11} \geq \bar E, \|{\bf w}_2\|^2 \!=\!1 \!\\\label{PartialCSIT2_3} & 0 \leq P_1 \leq P.
\end{eqnarray}
Because
\begin{eqnarray}\label{PartialCSIT3}
\log (1 + P{\bf w}_2^H{\bf H}_{22}^H{\bf
R}_{-2}^{-1}{\bf H}_{22}{\bf w}_2 )= \nonumber\\ \log (1 +
\frac{P{\bf w}_2^H{\bf H}_{22}^H{\bf H}_{22}{\bf w}_2}{1 + P_1{\bf v}_G^H(\theta_1){\bf H}_{21}^H{\bf H}_{21}{\bf v}_G(\theta_1)} ),
\end{eqnarray}
by letting $\alpha \triangleq {\bf v}_G^H(\theta_1){\bf H}_{21}^H{\bf H}_{21}{\bf v}_G(\theta_1)$, (P2) is equivalent with
\begin{eqnarray}\label{PartialCSIT4_1}
 \underset{P_1, {\bf w}_2}{\text{ maximize}}&\frac{P{\bf w}_2^H{\bf H}_{22}^H{\bf H}_{22}{\bf w}_2}{1 + \alpha P_1}
\\\label{PartialCSIT4_2} \!\!{\text{subject to}}\!&\!\!P{\bf w}_2^H{\bf H}_{12}^H{\bf
H}_{12} {\bf w}_2 \!+\!
\omega_1 P_1 \geq \bar E, \|{\bf w}_2\|^2 \!=\!1 \!\\\label{PartialCSIT4_3} & 0 \leq P_1 \leq P.
\end{eqnarray}
By introducing a new variable,
\begin{eqnarray}\label{PartialCSIT4_1_revise}
&E_h \triangleq \omega_1 P_1 + P{\bf w}_2^H{\bf H}_{12}^H{\bf
H}_{12} {\bf w}_2 \nonumber&\\&({\text{ equivalently, }}P_1 =\frac{1}{\omega_1}(E_h - P{\bf w}_2^H{\bf H}_{12}^H{\bf
H}_{12} {\bf w}_2)),&
\end{eqnarray}
we have
\begin{eqnarray}\label{PartialCSIT5_1}
\!&\!(\!P2a\!)\! \underset{E_h, {\bf w}_2}{\text{ maximize}}\quad\frac{P{\bf w}_2^H{\bf H}_{22}^H{\bf H}_{22}{\bf w}_2}{1 + \frac{\alpha}{\omega_1}\left(E_h - P{\bf w}_2^H{\bf H}_{12}^H{\bf H}_{12}{\bf w}_2\right)}
&\!\\\!&\label{PartialCSIT5_2} \!\!{\text{subject to}}\quad\!\! E_h \geq \bar E ,\|{\bf w}_2\|^2 =1\!&\\\label{PartialCSIT5_3} &  P{\bf w}_2^H{\bf H}_{12}^H{\bf H}_{12}{\bf w}_2 \leq E_h \leq \omega_1 P + P{\bf w}_2^H{\bf H}_{12}^H{\bf
H}_{12} {\bf w}_2.\!&\!
\end{eqnarray}
Note that the objective function is monotonic decreasing with respect to $E_h$. Therefore, to maximize (\ref{PartialCSIT5_1}) with respect to $E_h$, $E_h$ can be replaced by its lower bound in (\ref{PartialCSIT5_1}). From (\ref{PartialCSIT5_2}) and the first inequality of (\ref{PartialCSIT5_3}), when $P{\bf w}_2^H{\bf H}_{12}^H{\bf H}_{12}{\bf w}_2 > \bar E$, the constraint of (\ref{PartialCSIT5_2}) becomes inactive. Then, by substituting $E_h$ with $P{\bf w}_2^H{\bf H}_{12}^H{\bf H}_{12}{\bf w}_2 $ in (\ref{PartialCSIT5_1}), the optimal solution $\bar{\bf w}_2$ of (P2a) becomes an eigen-beamforming on ${\bf H}_{22}$, given as
\begin{eqnarray}\label{PartialCSIT5_5}
\bar{\bf w}_2 = {\bf w}_I, \quad {\bf w}_I = [{\bf V}_{22}]_1,
\end{eqnarray}
where ${\bf V}_{22}$ is an $M\times M$ unitary matrix form the SVD of ${\bf H}_{22}$. The corresponding $P_1$ is equal to $0$, which implies that the energy harvested from the information transmitter is enough to satisfy the target energy $\bar E$ and the energy transmitter does not transmit any signal, therefore not causing any interference to the ID MS.
Next, when $P{\bf w}_2^H{\bf H}_{12}^H{\bf H}_{12}{\bf w}_2 \leq \bar E$, 
the lower bound of (\ref{PartialCSIT5_3}) becomes inactive and by substituting $E_h$ with its lower bound $\bar E$ in (\ref{PartialCSIT5_1}), (P2a) can be rewritten as
\begin{eqnarray}\label{PartialCSIT6_1}
\!\!(\!P2b\!)\! \underset{{\bf w}_2}{\text{ maximize}}&\frac{P{\bf w}_2^H{\bf H}_{22}^H{\bf H}_{22}{\bf w}_2}{1 + \frac{\alpha}{\omega_1}\left(\bar E - P{\bf w}_2^H{\bf H}_{12}^H{\bf H}_{12}{\bf w}_2\right)}
\\\label{PartialCSIT6_2} \!\!{\text{subject to}}\!&\!\! \|{\bf w}_2\|^2 =1\!\\\label{PartialCSIT6_3} &  \bar E - \omega_1 P \leq P{\bf w}_2^H{\bf H}_{12}^H{\bf
H}_{12} {\bf w}_2 \leq \bar E.
\end{eqnarray}
Because ${\bf w}_2^H {\bf Q}{\bf w}_2 = tr ({\bf Q}{\bf W}_2)$ for any matrix ${\bf Q}$ with ${\bf W}_2 = {\bf w}_2 {\bf w}_2^H$, by relaxing the rank constraint of ${\bf W}_2$, we can have the following SDP relaxation problem for (P2b) as
\begin{eqnarray}\label{PartialCSIT7_1}
\!&\!\!(\!P2c\!)\! \underset{{\bf W}_2\succeq {\bf 0}}{\text{ maximize}}\quad\frac{P tr({\bf H}_{22}^H{\bf H}_{22}{\bf W}_2)}{tr \left(\left( (1+ \frac{\alpha}{\omega_1}\bar E) {\bf I} - \frac{\alpha}{\omega_1}P{\bf H}_{12}^H{\bf H}_{12}\right){\bf W}_2\right)}
\!&\!\\\!&\!\label{PartialCSIT7_2} \!\!{\text{subject to}}\!\quad\! tr({\bf W}_2) = 1\!&\!\\\!\!&\!\!\label{PartialCSIT7_3}  P tr({\bf H}_{12}^H{\bf
H}_{12} {\bf W}_2) \!\geq\! \bar E\! -\! \omega_1 P,~ P tr({\bf H}_{12}^H{\bf
H}_{12} {\bf W}_2) \!\leq\! \bar E.\!\!&\!\!
\end{eqnarray}
Note that the objective function in (P2c) is quasi-linear and it can be transformed into a linear program \cite{SBoyd}. That is, by introducing new variables ${\bf W}'_2 =\frac{ {\bf W}_2}{tr\left(\left( (1+ \frac{\alpha}{\omega_1}\bar E) {\bf I} - \frac{\alpha}{\omega_1}P{\bf H}_{12}^H{\bf H}_{12}\right){\bf W}_2 \right) }$ and $z = \frac{1}{tr\left(\left( (1+ \frac{\alpha}{\omega_1}\bar E) {\bf I} - \frac{\alpha}{\omega_1}P{\bf H}_{12}^H{\bf H}_{12}\right){\bf W}_2 \right)}$, (P2c) can be transformed into
\begin{eqnarray}\label{PartialCSIT8_1}
\!&\!\!(\!P2d\!) \underset{{\bf W}'_2\succeq {\bf 0}, z \geq 0}{\text{ maximize}}\quad  P tr({\bf H}_{22}^H{\bf H}_{22}{\bf W}'_2)
\!&\!\\\!\!&\!\label{PartialCSIT8_2} \!{\text{subject to}}\!\!\quad\!\! tr( \left( (1\!+ \!\frac{\alpha}{\omega_1}\bar E) {\bf I} - \frac{\alpha}{\omega_1}P{\bf H}_{12}^H{\bf H}_{12}\right) {\bf W}'_2) \!= \!1\!\!&\!\!\\\!&\!\label{PartialCSIT8_3} \! P tr({\bf H}_{12}^H{\bf
H}_{12} {\bf W}'_2) \geq (\bar E - \omega_1 P) z,\!&\!\\\!\!&\!\!\label{PartialCSIT8_4} P tr({\bf H}_{12}^H{\bf
H}_{12} {\bf W}'_2) \!\leq\! \bar E z,~   tr({\bf W}'_2) - z =0.\!\!&\!\!
\end{eqnarray}
Since Problem (P2d) is convex and satisfies the Slater's condition \cite{SBoyd}, it has a zero duality gap and its Lagrangian function is given as:
\begin{eqnarray}\label{PartialCSIT9_1}
\!\!&\!\!L({\bf W}'_2, z, \lambda, \mu)\! = \! P tr({\bf H}_{22}^H{\bf H}_{22}{\bf W}'_2)+\lambda \bigl(P tr({\bf H}_{12}^H{\bf
H}_{12} {\bf W}'_2)\nonumber\!\!&\!\!\\\!\!&\!\!-(\bar E - \omega_1 P) z\bigr)- \mu \left(P tr({\bf H}_{12}^H{\bf
H}_{12} {\bf W}'_2)- \bar E  z\right) \!\!&\!\!\\\!\!&\!\!= P tr({\bf A}{\bf W}'_2)+ \left(\mu\bar E -\lambda(\bar E - \omega_1 P)  \right)z,\nonumber\!\!&\!\!
\end{eqnarray}
where
 \begin{eqnarray}\label{PartialCSIT9_3}
{\bf A} = {\bf H}_{22}^H{\bf H}_{22} + (\lambda - \mu){\bf H}_{12}^H{\bf H}_{12}.
\end{eqnarray}
Then, the optimal $\bar{\bf W}'_2$ can be obtained by solving the dual problem of (P2d) as $\underset{\lambda, \mu \geq 0}{\min}\underset{{\bf W}'_2\succeq{\bf 0},z\geq 0}{\max}L({\bf W}'_2, z, \lambda, \mu)$ and is given as:
\begin{eqnarray}\label{PartialCSIT10}
\bar{\bf W}'_2 = \bar{\bf w}'_2 \bar{\bf w}_2'^H,~~ \bar{\bf w}'_2 = \frac{1}{\beta}[{\bf U}_{\bf A}]_1,
\end{eqnarray}
where ${\bf U}_{\bf A}$ is a unitary matrix from the EVD of ${\bf A}$ and $\beta$ is a scale factor such that the constraint (\ref{PartialCSIT8_2}) is satisfied. The corresponding $\bar\lambda$ and $\bar\mu$ can also be obtained by using the subgradient-based
method \cite{Zhang1, XZhao}, where the the subgradient is given by $\!(\left(\!P tr({\bf H}_{12}^H{\bf
H}_{12} {\bf W}'_2)\!-\!(\bar E \!-\! \omega_1 P) z\!\right)\!, \left(\!\bar E  z \!- \!P tr({\bf H}_{12}^H{\bf
H}_{12} {\bf W}'_2)\!\right) ) $ with $z =tr({\bf W}'_2)$. That is, the optimal $\bar{\bf W}'_2$ together with $\bar \lambda$ and $\bar \mu$ can be iteratively computed. Because (P2d) is convex, the solution in (\ref{PartialCSIT10}) is globally optimal under the local CSIT with the energy/information beamformers. Then, the optimal $\bar{\bf W}_2$ for (P2c) can be simply computed as:
\begin{eqnarray}\label{PartialCSIT11}
\bar{\bf W}_2 = \bar{\bf w}_2 \bar{\bf w}_2^H,~~ \bar{\bf w}_2 = \frac{1}{\|\bar{\bf w}'_2 \|}\bar{\bf w}'_2.
\end{eqnarray}
Note that because the optimal solution for (P2c) has a rank equal to one, (\ref{PartialCSIT11}) is also optimal for (P2b), (P2a), and (P2) without rank-relaxation. Here, $P_1$ can be determined as $P_1 =\frac{1}{\omega_1}(\bar E - P\bar{\bf w}_2^H{\bf H}_{12}^H{\bf
H}_{12} \bar{\bf w}_2)$. Accordingly, the iterative algorithm for the Geodesic energy beamforming and rank-one information beamforming can be summarized in Algorithm 3.

\vspace*{2pt}Algo. 3. {\it{\underline{Iterative algorithm for the
achievable R-E region}}}\\{\it{\underline{for Geodesic energy beamforming and rank-one information}}}
\\{\it{\underline{beamforming:}}}
\begin{enumerate}
\item Compute $\bar{\bf w}_2$ as in (\ref{PartialCSIT5_5}). If $P\bar{\bf w}_2^H{\bf H}_{12}^H{\bf H}_{12}\bar{\bf w}_2 \geq \bar E$, set $P_1=0$ and terminate the algorithm. Else, initialize $n=0$ and determine $\phi_0^{(0)}$ that minimizes $|P \|{\bf H}_{11} {\bf v}_G(\phi_0^{(0)}) \|^2 - \bar E|$.
\item For $n=0:N_{max}$
\!\!\!\begin{enumerate}
\item Find $\theta_1^{(n)}\in [0, \phi_0^{(n)}]$ as (\ref{prop2_1}), update ${\bf v}_G(\theta_1^{(n)})$, and solve (P2d), resulting in ${\bf Q}_2^{(n)} = P \bar{\bf w}_2^{(n)} (\bar{\bf w}_2^{(n)})^H$.
\item Then, update $\phi_0^{(n)}$ such that
\begin{eqnarray}\label{ReviseAlgoGeo3_2}
P \|{\bf H}_{11} {\bf v}_G(\phi_0^{(n)}) \|^2 = \bar E - {\bf H}_{12}{\bf Q}_2^{(n)}{\bf H}_{12}^H.
\end{eqnarray}
\end{enumerate}
\item Finally, determine the energy transmit power as
$P_1 =\frac{1}{\|{\bf H}_{11} {\bf v}_G(\theta_1^{(n)}) \|^2}(\max\{\bar E - P\|{\bf
H}_{12} \bar{\bf w}_2\|^2,0\})$ and the boundary point of the achievable R-E region is given as
\begin{eqnarray}\label{ReviseAlgoGeo3_5}
(R, E) =(\log \det
\left({\bf I}_{M} + \tilde{\bf H}_{22}{\bf Q}_2^{(N_{max})}\tilde{\bf
H}_{22}^H \right),\nonumber\\~\|{\bf H}_{11} {\bf v}_G(\theta_1^{(N_{max})}) \|^2 + tr \left({\bf H}_{12} {\bf
Q}_2^{(N_{max})} {\bf H}_{12}^H\right)).
\end{eqnarray}
\end{enumerate}
 \vspace*{2pt}
Even though $P_1$ is computed once in Step 3 of Algorithm 3 (c.f., in Step 2 of Algorithm 2, $P_1$ is iteratively updated), (P2d) in Step 2.a is actually optimized with respect to $P_1$, implicitly, because we have replaced $P_1$ as (\ref{PartialCSIT4_1_revise}).
\begin{remark}\label{remark4}
Note that when $P{\bf w}_2'^H{\bf H}_{12}^H{\bf H}_{12}{\bf w}'_2 \leq \bar E z$, from (\ref{PartialCSIT9_1}), $\mu$ minimizing $\underset{{\bf W}'_2\succeq{\bf 0},z\geq 0}{\max}L({\bf W}'_2, z, \lambda, \mu)$  will be zero. In addition, when $\bar E$ is small enough such that $ P tr({\bf H}_{12}^H{\bf
H}_{12} {\bf W}'_2) \approx \bar E z$, the subgradient of $\lambda$ is positive and therefore, the value of $\bar \lambda$ minimizing $\underset{{\bf W}'_2\succeq{\bf 0},z\geq 0}{\max}L({\bf W}'_2, z, \lambda, \mu)$ approaches 0. That is, the optimal $\bar {\bf w}_2$ approaches ${\bf w}_I$ in (\ref{PartialCSIT5_5}). In contrast, when $\bar E$ is large resulting in the subgradient of $\lambda$ being negative, the value of $\bar \lambda$ will increase. That is, the solution approaches
\begin{eqnarray}\label{PartialCSIT10_1}
\bar{\bf w}_2 = {\bf w}_L, \quad {\bf w}_L= [{\bf V}_{12}]_1,
\end{eqnarray}
where ${\bf V}_{12}$ is an $M\times M$ unitary matrix from the SVD of ${\bf H}_{12}$. That is, the optimal $\bar{\bf w}_2$ will approach the beamforming vector such that the energy transferred through ${\bf H}_{12}$ is maximized.
\end{remark}


\subsection{Geodesic based rank-one energy/information beamforming}\label{ssec:bothGeodesic}
Motivated by Remark \ref{remark4}, we can define a Geodesic information beamforming vector ${\bf w}_G(\theta_2)$ with $[{\bf V}_{22}]_1(\triangleq {\bf w}_I)$ and $[{\bf V}_{12}]_1(\triangleq {\bf w}_L)$ as:
\begin{eqnarray}\label{bothG1}
{\bf w}_G(\theta_2) = {\bf w}_I {u}_I \cos(\theta_2) - ({\bf w}_I)^{\perp}{\bf u}_L \sin(\theta_2),
\end{eqnarray}
where $({\bf w}_I)^{\perp}{\bf u}_L =[{\bf w}_I u_I \cos(\phi_I) - {\bf w}_L]^{-1}\sin(\phi_I)^{-1} $. Here, ${\phi}_I$ and ${u}_I$ are the principle angle and the phase difference between $ {\bf w}_I$ and ${\bf w}_L$, respectively, such that $ {\bf w}_I^H{\bf w}_L= u_I\cos{\phi}_I$. Then, we have the following theorem.

\begin{thm}\label{thm2} When the information transmitter opts for the rank-one beamforming, the optimal beamforming vector ${\bf w}_2$ lies in the Geodesic curve between $[{\bf V}_{22}]_1$ and $[{\bf V}_{12}]_1$.
\end{thm}
\begin{proof}
See Appendix \ref{appndix4}.
\end{proof}

Accordingly, to evaluate the achievable region, we optimize $P_1$, $\theta_1$ and ${\theta}_2$ in a distributed manner. Therefore, $\theta_1$ can be determined such as (\ref{prop2_1}) and $\omega_1 = {\bf v}_G^H(\theta_1){\bf H}_{11}^H{\bf H}_{11}{\bf v}_G(\theta_1)$. Then, by substituting $P_1 =\max\{\frac{1}{\omega_1}(\bar E - P{\bf w}_G^H(\theta_2){\bf H}_{12}^H{\bf
H}_{12} {\bf w}_G(\theta_2)), 0\}$ from (\ref{PartialCSIT4_1_revise}) and (\ref{PartialCSIT4_2}) into (\ref{PartialCSIT4_1}), we can find the optimal $\theta_2$ such that
\begin{eqnarray}\label{bothG2}
\!\!\underset{0\!\leq \!\theta_2\!\leq\! \phi_I}{\max} \!J(\!\theta_2\!) \!=\! \frac{P{\bf w}_G^H(\theta_2){\bf H}_{22}^H{\bf H}_{22}{\bf w}_G(\theta_2)}{1\! +\!  \frac{\alpha}{\omega_1}(\max\{\bar E \!-\! P{\bf w}_G^H(\!\theta_2\!){\bf H}_{12}^H{\bf
H}_{12} {\bf w}_G(\!\theta_2\!),0\})}\!,\!\!
\end{eqnarray}
The iterative algorithm for the Geodesic energy/information beamforming can be summarized in Algorithm 4.

%
%

\vspace*{2pt}Algo. 4. {\it{\underline{Iterative algorithm for the
achievable R-E region}}}\\{\it{\underline{for Geodesic energy/information beamforming:}}}
\begin{enumerate}
\item If $P{\bf w}_G(0)^H{\bf H}_{12}^H{\bf H}_{12}{\bf w}_2(0) \geq \bar E$, set $P_1=0$ and terminate the algorithm. Else, initialize $n=0$ and determine $\phi_0^{(0)}$ that minimizes $|P \|{\bf H}_{11}{\bf v}_G(\phi_0^{(0)})\|^2 - \bar E|$.
\item For $n=0:N_{max}$
\!\!\!\begin{enumerate}
\item Find $\theta_1^{(n)}\in [0, \phi_0^{(n)}]$ as (\ref{prop2_1}) and solve (\ref{bothG2}) for $\theta_2^{(n)}$.
\item Then, update $\phi_0^{(n)}$ such that $P \|{\bf H}_{11}{\bf v}_G(\phi_0^{(n)})\|^2 = \bar E - P\|{\bf H}_{12}{\bf w}_G(\theta_2^{(n)})\|^2$.
\end{enumerate}
\item Finally, determine the energy transmit power as
$P_1 =\frac{1}{\|{\bf H}_{11} {\bf v}_G(\theta_1^{(n)}) \|^2}(\max\{\bar E - P\|{\bf
H}_{12} {\bf w}_G(\theta_2^{(n)})\|^2, 0\})$ and the boundary point of the achievable R-E region is given as
\begin{eqnarray}\label{ReviseAlgoGeo2_5}
(R, E) =(\log \det
\left({\bf I}_{M} + \tilde{\bf H}_{22}{\bf Q}_2^{(N_{max})}\tilde{\bf
H}_{22}^H \right),\nonumber\\~ E_{11}^{(N_{max})} + tr \left({\bf H}_{12} {\bf
Q}_2^{(N_{max})} {\bf H}_{12}^H\right)).
\end{eqnarray}
\end{enumerate}
 \vspace*{2pt}

\begin{remark}\label{remark5}
In algorithm 4, the information transmitter does not need the full CSI of their channels to both energy/information receivers, but requires ${\bf w}_I$ and ${\bf w}_L$. Similarly to Remark \ref{remark3}, because $J(\theta_2)$ in (\ref{bothG2}) can be rewritten as shown at the top of the next page,
\begin{figure*}[!t]
\scriptsize
\begin{eqnarray}\label{eqn_remark5_1}
J(\theta_2) = \frac{P(  \cos^2(\theta_2)\sigma_{22,1}^2 +\sin^2(\theta_2) \| {\bf H}_{22}({\bf w}_I)^{\perp}{\bf u}_L\|^2 )}{1 +  \frac{\alpha}{\omega_1}(\max\{\bar E - P( \cos^2(\phi_I - \theta_2)\sigma_{12,1}^2 +\sin^2(\phi_I - \theta_2) \| {\bf H}_{12}({\bf w}_L)^{\perp}{\bf u}_I\|^2 ) ,0\})},\nonumber
\end{eqnarray}
\hrulefill \vspace*{2pt}
\end{figure*}
to evaluate $J(\theta_2)$, four scalar values of $\sigma_{22,1}^2 $, $\sigma_{12,1}^2$, $ \| {\bf H}_{22}({\bf w}_I)^{\perp}{\bf u}_L\|^2$, and $ \| {\bf H}_{12}({\bf w}_L)^{\perp}{\bf u}_I\|^2$ are additionally required, where
$({\bf w}_L)^{\perp}{\bf u}_I =[{\bf w}_L u_L \cos(\phi_I) - {\bf w}_I]^{-1}\sin(\phi_I)^{-1} $ and $u_L = u_I^H$. The last two of them can be evaluated at each MS from two different reference signals (with $({\bf w}_I)^{\perp}{\bf u}_L$ and $({\bf w}_L)^{\perp}{\bf u}_I$, respectively) of the information transmitter and reported back to the information transmitter. Note that the information transmitter does not require the information of the interference covariance matrix ${\bf R}_{-2}^{(n)}$ in (\ref{ReviseAlgoGeo2_1}) of Algorithm 2 but requires two scalar values $\omega_1$ and $\alpha$, because it only needs to find ${\theta}_{2}^{(n)}$ for information beamforming in (\ref{bothG2}). Note that Algorithms 3 and 4 exhibit the same optimal R-E region under the local CSIT with the non-cooperative energy/information beamformers (See also Fig. 5), but Algorithm 4 benefits from a further reduced feedback overhead compared to Algorithm 3.
\end{remark}

\subsection{Discussion: Extension to K-user MIMO IFC}\label{ssec:KuserIFC}
In \cite{ParkClerckx2}, we have shown that JWIET problem in the K-user MIMO IFC can be transformed into an equivalent two-user MIMO IFC with additional constraints (the covariance matrix of external interferences at the effective ID MS and the block diagonal constraints on the covariance matrix of the effective information transmitter, see also \cite{ParkClerckx2}) and the optimal energy beamforming strategy has a rank-one beamforming. Therefore, the Geodesic beamforming can be extended to general K-user MIMO IFC. That is, if each energy transmitter can find two optimal directions such that either the system energy is maximized (energy maximum direction, EMD) or the interference is minimized (interference minimum direction, IMD), it can steer the beam lying on the Geodesic curve between the EMD and IMD vectors. However, finding the EMD and IMD vectors at each transmitter requires full local CSIT of its associated channel links. To extend our proposed Geodesic beamforming scheme (with a partial feedback of the unit-norm singular vectors, i.e., either $[{\bf V}_{ij}]_1$ or $[{\bf V}_{ij}]_M$), EAP and IAP have to estimate the EMD and IMD based on those partial feedback information. Let us assume that we have $K_1$ energy transceiver pairs and $K-K_1$ information transceiver pairs and ${\bf H}_{ij}$ is denoted as the channel from the $i$th transmitter to the $j$th receiver. Without loss of generality, the $i$th MS, $i=1,...,K_1$ harvests the energy. Then, to maximize the transferred energy to EH MSs, each transmitter should find or estimate the singular vector associated with the largest singular value (simply, largest singular vector) of $[{\bf H}_{1j}^T,...,{\bf H}_{K_1 j}^T]^T$, $j=1,...,K$. If the $i$th EH MS reports the largest singular value and the associated singular vector of each channel matrix from the transmitters, respectively, i.e., ${\sigma}_{ij,1}$ and $[{\bf V}_{ij}]_1$ for $j=1,...,K$, one simple approach to estimate the largest singular vector (or, EMD) based on the partial CSIT at the $j$th AP is the selection method such as
\begin{eqnarray}\label{Discusseqn1}
{\bf v}_{j,EMD} = [{\bf V}_{\bar ij}]_1 {\text{ such that }} \bar i = \underset{i=1,...,K_1}{\arg}{\max} ~{\sigma}_{i j,1}.
\end{eqnarray}
Or, we can compute ${\bf v}_{j,EMD}$ as
\begin{eqnarray}\label{Discusseqn2}
{\bf v}_{j,EMD} = [\bar{\bf U}_j]_1,
\end{eqnarray}
where $[\bar{\bf U}_j]_1$ is the largest left singular vector of $ \left[[{\bf V}_{1j}]_1,...,[{\bf V}_{K_1j}]_1 \right]diag\{{\sigma}_{1j,1},...,{\sigma}_{K_1j,1} \}$. That is, ${\bf v}_{j,EMD}$ is the largest singular vector of the range space of $[{\bf V}_{ij}]_1$, $i=1,...,K_1$. 

In the ID MSs, the signals via all the cross links are the interference signals. Therefore, each ID MS reports the largest singular vector, $[{\bf V}_{ij}]_1$ ($i=j$) to the serving IAP and the minimum singular vector, $[{\bf V}_{ij}]_M$, to the other IAPs and EAPs ($i\neq j$). The feedback strategy is described in Fig. \ref{JWIET_KuserIC_feedback}. Then, the transmitters can then estimate IMD vectors based on the partial CSIT, similarly to (\ref{Discusseqn1}) and (\ref{Discusseqn2}). Once EMD and IMD vectors are estimated based on the partial CSIT, we can optimize $\theta_i$ and the transmit power of EAP in a distributed way to satisfy the target harvesting energy. That is, we set $\theta_i = 0$ and the transmit power of EAP as maximum. If the harvested energy is larger than the target energy, then each EAP tilts beams by increasing $\theta_i$ and simultaneously reduces its power $P_i$ to decrease the interference to ID MSs, until the harvested energy meets the target energy \cite{ParkClerckx2}.

%

\begin{figure}
\begin{center}
\begin{tabular}{c}
\includegraphics[height=4.2cm]{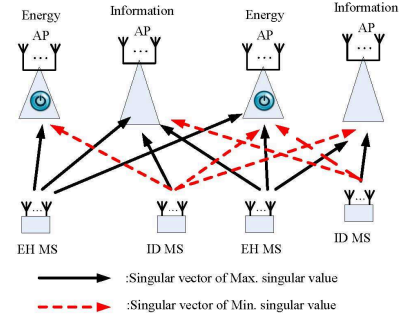}
\end{tabular}
\end{center}
\caption[JWIET_KuserIC_feedback]
{ \label{JWIET_KuserIC_feedback} Feedback strategy for general K-user MIMO IFC.}
\end{figure}

\section{Adaptive feedback bit allocation strategy for Geodesic based energy/information beamforming}\label{sec:Bitallocation}
Because ${\bf v}_E$, ${\bf v}_L$, ${\bf w}_I$, and ${\bf w}_L$ are i.i.d. isotropically distributed on $M$ dimensional unit-norm sphere, to report them to their respective transmitters, EH MS (resp. ID MS) can utilize the RVQ with the codebooks $\mathfrak{C}_{i1}\triangleq \{{\bf f}_l^{i1},l=1,..., 2^{B_{i1}}\}$ (resp. $\mathfrak{C}_{i2}\triangleq \{{\bf f}_l^{i2},l=1,..., 2^{B_{i2}}\}$), $i=1,2$ as
\begin{eqnarray}\label{RVQ_1}
&\hat{\bf v}_E= \arg \underset{{\bf f}_l^{11},l=1,..., 2^{B_{11}}}{\max} |{\bf v}_E^H {\bf f}_{l}^{11}|,\nonumber&\\& \hat{\bf w}_L = \arg \underset{{\bf f}_l^{21},l=1,..., 2^{B_{21}}}{\max} |{\bf w}_L^H {\bf f}_l^{21}|, {\text{ at EH MS}},\nonumber&\\&
\hat{\bf v}_L= \arg \underset{{\bf f}_l^{12},l=1,..., 2^{B_{12}}}{\max} |{\bf v}_L^H {\bf f}_{l}^{12}|,\nonumber&\\& \hat{\bf w}_I = \arg \underset{{\bf f}_l^{22},l=1,..., 2^{B_{22}}}{\max} |{\bf w}_I^H {\bf f}_l^{22}|, {\text{ at ID MS}},&
\end{eqnarray}
where $B_{ij}$ is the number of feedback bits that is reported by the $j$th MS to the $i$th AP and $B_{11} + B_{21} = B_{12}+ B_{22} = B$. Assuming that the quantized information is perfectly reported to both energy/information transmitters with zero-delay, from (\ref{Partial5}) and (\ref{bothG1}), the estimated Geodesic energy/information beamforming can be given as
\begin{eqnarray}\label{RVQ_2}\nonumber
\!\!\!\hat{\bf v}_G(\theta_1) \!&\!\!=\!\!& \!\hat{\bf v}_E \hat{u}_E\! \cos(\theta_1)\! \\\nonumber\!&\!\!\!\!&-\! [\hat{\bf v}_E \hat{u}_E\! \cos(\hat\phi_E) \!-\! \hat{\bf v}_L](\sin\hat\phi_E)^{-\!1}\! \sin(\theta_1)\!,\!\!\\\!\!\!
\hat{\bf w}_G(\theta_2) \!&\!\!=\!\!&\! \hat{\bf w}_I \hat{u}_I \!\cos(\theta_2) \!\nonumber\\&\!\!\!\!&-\! [\hat{\bf w}_I \hat{u}_I \!\cos(\hat\phi_I)\! -\! \hat{\bf w}_L](\sin\hat\phi_I)^{-\!1}\! \sin(\theta_2)\!,\!\!\nonumber
\end{eqnarray}
where $ \hat\phi_E =\cos^{-1}|\hat{\bf v}_E^H \hat{\bf v}_L| $ and $ \hat\phi_I =\cos^{-1}|\hat{\bf w}_I^H \hat{\bf w}_L| $. Then, the following proposition is useful to develop the adaptive feedback bit allocation strategy for the EH MS.
\begin{prop}\label{prop3_rev} The energy transferred from the energy/information transmitters is lower bounded as shown at the top of the next page.
\begin{figure*}[!t]
\scriptsize
\begin{eqnarray}\label{prop3_1}
E\left[P_1\|{\bf H}_{11}\hat{\bf v}_G(\theta_1)\|^2\right] \!\!&\!\!\geq\!\! &\ E\left[\|{\bf H}_{11}\hat{\bf v}_E\|^2\right] P_1\left(\cos^2\theta_1 - \sin^2\theta_1\left(1- B(1, \frac{M}{M-1})\right) \right) +MP_1 \sin^2\theta_1 \triangleq E_{11}^{low},
\end{eqnarray}
\begin{eqnarray}\label{prop3_1_1}
E\left[P\|{\bf H}_{12}\hat{\bf w}_G(\theta_2)\|^2\right] \!\!&\!\!\geq\!\! &\!E\left[\|{\bf H}_{12}\hat{\bf w}_L\|^2\right]P \left(\cos^2(\hat\phi_I -\theta_2) - \sin^2(\hat\phi_I -\theta_2)\left(1- B(1, \frac{M}{M-1})\right) \right)+\alpha_{12}PM \sin^2(\hat\phi_I -\theta_2) \triangleq E_{12}^{low},
\end{eqnarray}
where $B(x, y)$ is the Beta function.

\hrulefill \vspace*{2pt}
\end{figure*}
Here, $ 0\leq \theta_1 \leq \hat\phi_E$ and $ 0\leq \theta_2 \leq \hat\phi_I$, respectively. Note that $P_1$, $\theta_1$, and $\theta_2$ are the parameters dependent on the target harvesting energy $\bar E$.
%
\end{prop}
\begin{proof}
From (\ref{Prop1_9}) in Appendix A,
\begin{eqnarray}\label{prop3_2}
\!\!\!E\left[P_1\!\|{\bf H}_{11}\hat{\bf v}_G(\theta_1)\|^2\right] \!=\! P_1\!\cos^2(\theta_1)E\left[\|{\bf H}_{11}\hat{\bf v}_E\|^2\right]\!+\quad \!\! \nonumber\\\!\!\!P_1 \!\sin^2(\theta_1)\!\underbrace{ E\left[\| {\bf H}_{11}[\hat{\bf v}_E \hat{u}_E \cos(\hat\phi_E) \!-\! \hat{\bf v}_L](\sin\hat\phi_E)^{-1}\|^2\right]}_{(a)}\!,\!\!
\end{eqnarray}
where $ \hat u_E\cos\hat\phi_E =\hat{\bf v}_E ^H\hat{\bf v}_L$. Here, $\hat u_E \triangleq e^{j\hat\psi}$ is the phase difference between $\hat{\bf v}_E$ and $\hat{\bf v}_L$ and $\hat\psi$ is uniformly distributed on $[0, 2\pi]$ which implies that $E[\hat u_E]=0$. From \cite{N_Jindal2,CAuYeung}, we have
\begin{eqnarray}\label{prop3_3}
E\left[\|[{\bf V}_{11}]_1^H\hat{\bf v}_E \|^2\right] =1-2^{B_{11} }B(2^{B_{11} }, \frac{M}{M-1}).
\end{eqnarray}
In addition, from \cite{N_Jindal2}, $\hat{\bf v}_E$ can be modeled as $\hat{\bf v}_E = \sqrt{1-z^2}{\bf v}_E + z {\bf s}$, where ${\bf s}$ is a unit-norm vector isotropically
distributed in the null space of ${\bf v}_E$ and $z$ is quantization error with $E[z]= 2^{B_{11} }B(2^{B_{11} }, \frac{M}{M-1})$. Note that ${\bf H}_{11}$ is zero-mean normalized Gaussian distributed and independent with $\hat{\bf v}_{L}$, $E\left[ \|{\bf H}_{11}\hat{\bf v}_L\|^2\right] = M$. Furthermore, $\hat{\bf v}_L$ can be rewritten as
\begin{eqnarray}\label{prop3_4_rev1}
\hat{\bf v}_L = {\hat u}_E\cos\hat\phi_E \hat{\bf v}_E + {\bf P}_{ \hat{\bf v}_E}^{\perp}\hat{\bf v}_L,
\end{eqnarray}
where ${\bf P}_{ \hat{\bf v}_E}^{\perp} = {\bf I}_M - \hat{\bf v}_E \hat{\bf v}_E^H$ and ${\hat u}_E$ is independent with the second term in (\ref{prop3_4_rev1}) because ${\hat u}_E$ depends only on the inner product of $\hat{\bf v}_E$ and $\hat{\bf v}_L$. By substituting (\ref{prop3_4_rev1}) into the second expectation of (\ref{prop3_2}), it is then lower bounded as shown at the top of the next page.
\begin{figure*}[!t]
\scriptsize
\begin{eqnarray}\label{prop3_6}
\!\!(a)\!&\!\!=\! \!&\!\!E\left[\frac{1}{\sin^2\hat\phi_E}\left(\cos^2\hat\phi_E \|{\bf H}_{11}\hat{\bf v}_E\|^2 + \|{\bf H}_{11}\hat{\bf v}_L\|^2 \!-\!{\hat u}_E\cos\hat\phi_E \hat{\bf v}_L^H{\bf H}_{11}^H{\bf H}_{11}\hat{\bf v}_E\! -\!\left({\hat u}_E\cos\hat\phi_E \hat{\bf v}_L^H{\bf H}_{11}^H{\bf H}_{11}\hat{\bf v}_E\!\right)^H\!\right)\right]\!,\nonumber\!\!\\
\!\!&\!\!=\! \!&\!\! E\biggl[\! \frac{1}{\sin^2\hat\phi_E} \biggl(M\!-\! \cos^2\hat\phi_E\|{\bf H}_{11}\hat{\bf v}_E\|^2\!-\!{\hat u}_E\cos\hat\phi_E \hat{\bf v}_L^H{\bf P}_{ \hat{\bf v}_E}^{\perp}{\bf H}_{11}^H{\bf H}_{11}\hat{\bf v}_E\! -  \!\left(\!{\hat u}_E\cos\hat\phi_E \hat{\bf v}_L^H{\bf P}_{ \hat{\bf v}_E}^{\perp}{\bf H}_{11}^H{\bf H}_{11}\hat{\bf v}_E\!\right)^H\!\biggr)\biggr],\!\!\nonumber\\
\!\!&\!\!=\!\!&\!\! E\left[\frac{1}{1-\cos^2\hat\phi_E}( M- \cos^2\hat\phi_E  \|{\bf H}_{11}\hat{\bf v}_E\|^2 ) \right]\geq M - (1 - B(1, \frac{M}{M-1}))E\left[ \|{\bf H}_{11}\hat{\bf v}_E\|^2\right],
\end{eqnarray}
where the equality in (\ref{prop3_6}) is from that $\hat u_E$ is independent with $\hat\phi_E$ and $E[\hat u_E]=0$. In addition, the last inequality is from $E[\cos^2\hat\phi_E] = 1 - B(1, \frac{M}{M-1})$ and $\cos^2\hat\phi_E \leq 1$.

\hrulefill \vspace*{2pt}
\end{figure*}
Therefore, we can have (\ref{prop3_1}) and in a similar way, we can also derive (\ref{prop3_1_1}).
\end{proof}
Note that $\|{\bf H}_{11}\hat{\bf v}_E\|^2 = \sum_{i=1}^M \sigma_{11,i}^2 \|[{\bf V}_{11}]_i^H\hat{\bf v}_E \|^2$, where $\sigma_{11,i}^2$ and $ \|[{\bf V}_{11}]_i^H\hat{\bf v}_E \|^2$ are independent. Then, the first expectation of (\ref{prop3_2}) is given as
\begin{eqnarray}\label{RVQ_5}
E\left[\|{\bf H}_{11}\hat{\bf v}_E\|^2\right] = E[\sigma_{11,1}^2]\left(1-2^{B_{11} }B(2^{B_{11} }, \frac{M}{M-1})\right) \nonumber\\+  E[\sum_{i=2}^M \sigma_{11,i}^2 \|[{\bf V}_{11}]_i^H\hat{\bf v}_E \|^2].
\end{eqnarray}
That is, the lower bound in (\ref{prop3_1}) has a quite complicated form, but, thanks to Lemma 1 in \cite{WSantipach} and the asymptotic results for large $M$ such as \cite{WSantipach,WSantipach2, ATulino}\footnote{Note that in \cite{WSantipach} and \cite{ATulino}, the entries are i.i.d. Gaussian RVs with a zero-mean and a variance of $\frac{1}{M}$.}
\begin{eqnarray}\label{RVQ_6}
\!\!&\!2^{B_{11} }B(2^{B_{11} }, \frac{M}{M-1}) \longrightarrow 2^{- \frac{B_{11}}{M}},~
E[\sigma_{11,1}^2] \longrightarrow 4M,\nonumber\!&\!\!\\\!\!&\!\!
E[\sum_{i\!=\!2}^M \sigma_{11,i}^2 \|[{\bf V}_{11}]_i^H\hat{\bf v}_E \|^2] \longrightarrow 2^{-\! \frac{B_{11}}{M}}\int_{0}^{\infty}\lambda g_{{\bf H}_{11}^H{\bf H}_{11}}\!(\lambda)d\lambda,\nonumber\!\!&\!\!\\\!\!&\!\! \int_{0}^{\infty}\lambda g_{{\bf H}_{11}^H{\bf H}_{11}}(\lambda)d\lambda \longrightarrow M,\!\!&\!\!
\end{eqnarray}
where $g_{H_{11}^H{\bf H}_{11}}(\lambda)$ is a deterministic function given by \cite{ATulino}, it can be asymptotically approximated for large $M$ as
\begin{eqnarray}\label{RVQ_7}
E_{11}^{low} \approx M P_1\left[ (4- 3 \cdot2^{- \frac{B_{11}}{M}})  \cos^2\theta_1 +\sin^2\theta_1\right].
\end{eqnarray}
Note that, from (\ref{RVQ_7}), as the number of antennas or feedback bits increase or $\theta_1$ decreases, the energy transferred from the energy transmitter will increase. Interestingly, when the number of feedback bits is zero, the transferred energy becomes $M$, independent with $\theta_1$. Similarly, the lower bound in (\ref{prop3_1_1}) can be approximated as
\begin{eqnarray}\label{RVQ_8}
\!\!E_{12}^{low} \!\approx \!\alpha_{12}MP\!\left[ (4\!- \!3 \cdot2^{- \!\frac{B_{21}}{M}})  \cos^2(\hat\phi_I \!-\! \theta_2) + \sin^2(\hat\phi_I \!-\!\theta_2)\right]\!,\nonumber\!
\end{eqnarray}
\begin{remark}\label{remark6}
To find the optimal $B_{11}$ and $B_{21}$, by substituting $B_{21} = B - B_{11}$ into $E_{11}^{low} +E_{12}^{low}$, we can search $B_{11}$ maximizing it, numerically. Fortunately, because $E_{11}^{low} +E_{12}^{low}$ is logarithmically concave, by computing $\nabla_{B_{11}}(E_{11}^{low} +E_{12}^{low}) = 0$, we get an optimal solution as
\begin{eqnarray}\label{RVQ_9}
\!\!B_{11}\! =\! min\!\left\{\! B, \left(\left\lfloor \frac{B}{2} \!+ \!\frac{M}{2}\log\frac{P_1\cos^2\theta_1}{\alpha_{12}P\cos^2(\hat\phi_I \!-\!\theta_2)} \right\rceil \right)^+\!\right\}\!.\!\!
\end{eqnarray}
That is, when the path loss of the cross link becomes large (or, $\alpha_{12}$ becomes small), $B_{11}$ should be increased. In addition, when $\theta_1$ is small (i.e., the energy transferred from the energy transmitter is large), $B_{11}$ should be increased. In contrast, when $P_1= 0$ (i.e., the harvested energy from the information transmitter is enough), $B_{11}=0$ which implies that all the feedback bits are allocated for the cross link.
\end{remark}

The following proposition is useful to develop the adaptive feedback bit allocation strategy for the ID MS.
\begin{prop}\label{prop4} The channel gain of information link is lower bounded as (\ref{prop4_1_1}),
while the interference from the energy transmitter is upper bounded as (\ref{prop4_1}) shown at the top of the next page.
\begin{figure*}[!t]
\scriptsize
\begin{eqnarray}\label{prop4_1_1}
\!E\left[P\|{\bf H}_{22}\hat{\bf w}_G(\theta_2)\|^2\right] \!&\!\geq\!&\! E\left[\|{\bf H}_{22}\hat{\bf w}_I\|^2\right]P \left(\cos^2(\theta_2) - \sin^2(\theta_2)\left(1- B(1, \frac{M}{M-1})\right) \right)+M P\sin^2(\theta_2) \triangleq S_{22}^{low},
\end{eqnarray}
\begin{eqnarray}\label{prop4_1}
E\left[P_1\|{\bf H}_{21}\hat{\bf v}_G(\theta_1)\|^2\right] \leq E\left[\|{\bf H}_{21}\hat{\bf v}_L\|^2\right]P_1\cos^2(\hat\phi_E-\theta_1) +\frac{\alpha_{21}MP_1}{B(1, \frac{M}{M-1})} \sin^2(\hat\phi_E-\theta_1) \triangleq IN_{21}^{up}.
\end{eqnarray}
\hrulefill \vspace*{2pt}
\end{figure*}
\end{prop}
\begin{proof}
Following a similar approach as Proposition \ref{prop3_rev}, (\ref{prop4_1_1}) can be easily derived. Note that
\begin{eqnarray}\label{prop4_2}
\!\!\!&\!\!E\!\left[P_1\!\|{\bf H}_{21}\hat{\bf v}_G(\theta_1)\|^2\right] \!= \!P_1\!\cos^2(\hat\phi_E\!-\!\theta_1)E\!\left[\|{\bf H}_{21}\hat{\bf v}_L\|^2\right] \!+ \!&\!\!\nonumber\\\!\!\!& \!\!P_1\!\sin^2(\hat\phi_E\!-\!\theta_1)\!\underbrace{ \!E\!\left[\| {\bf H}_{21}[\hat{\bf v}_L \hat{u}_L \cos(\hat\phi_E) \!-\! \hat{\bf v}_E](\sin\hat\phi_E)^{-\!1}\|^2\right]\!}_{(b)}\!\!.\!\!\!&\!\!\nonumber\\\!\!\!\!\!\!\!\!\!\!
\end{eqnarray}
From \cite{N_Jindal2,CAuYeung}, we again have
\begin{eqnarray}\label{prop4_3}
E\left[\|[{\bf V}_{21}]_M^H\hat{\bf v}_L \|^2\right] =1-2^{B_{12} }B(2^{B_{12} }, \frac{M}{M-1}),
\end{eqnarray}
and, from \cite{N_Jindal2}, $\hat{\bf v}_L$ can be modeled as $\hat{\bf v}_L = \sqrt{1-z^2}{\bf v}_L + z {\bf s}'$,
where ${\bf s}'$ is a unit-norm vector isotropically distributed in the null space of ${\bf v}_L$ and $z$ is quantization error with $E[z]= 2^{B_{12} }B(2^{B_{12} }, \frac{M}{M-1})$. Note that $\frac{1}{\alpha_{21}}{\bf H}_{21}$ is zero-mean normalized Gaussian distributed and independent with $\hat{\bf v}_{E}$, $E\left[ \|{\bf H}_{21}\hat{\bf v}_E\|^2\right]= \alpha_{21}M$. Furthermore, $\hat{\bf v}_E$ can be rewritten as
\begin{eqnarray}\label{prop4_4_rev1}
\hat{\bf v}_E = {\hat u}_L\cos\hat\phi_E \hat{\bf v}_L + {\bf P}_{ \hat{\bf v}_L}^{\perp}\hat{\bf v}_E,
\end{eqnarray}
where ${\bf P}_{ \hat{\bf v}_L}^{\perp}  = {\bf I}_M - \hat{\bf v}_L \hat{\bf v}_L^H$ and ${\hat u}_L = {\hat u}_E^H$ is independent with the second term in (\ref{prop4_4_rev1}). By substituting (\ref{prop4_4_rev1}) into the second expectation of (\ref{prop4_2}), it is then upper bounded as
\begin{eqnarray}\label{prop4_6}
\!\!(b)\!=\! E\!\left[\frac{1}{1\!-\!\cos^2\hat\phi_E}(\alpha_{21} M\!- \!\cos^2\hat\phi_E  \|{\bf H}_{21}\hat{\bf v}_L\|^2 ) \right]\nonumber\\ \leq \frac{\alpha_{21}M}{B(1, \frac{M}{M\!-\!1})},\!
\end{eqnarray}
where the last inequality is from the Jensen's inequality of $E\left[\frac{1}{1-x}\right] \leq \frac{1}{1-E[x]}$ with $x\in (0,1)$.
\end{proof}
Because $\|{\bf H}_{21}\hat{\bf v}_L\|^2 = \sum_{i=1}^M \sigma_{21,i}^2 \|[{\bf V}_{21}]_i^H\hat{\bf v}_L \|^2$, where $\sigma_{21,i}^2$ and $ \|[{\bf V}_{21}]_i^H\hat{\bf v}_L \|^2$ are independent, we have
\begin{eqnarray}\label{RVQ_10}
E\left[\|{\bf H}_{21}\hat{\bf v}_L\|^2\right]
= E[\sigma_{21,M}^2]E[\|[{\bf V}_{21}]_M^H\hat{\bf v}_L] \quad\quad\nonumber\\+ E[\sigma_{21,1}^2]E[\|[{\bf V}_{21}]_1^H\hat{\bf v}_L]+   E[\sum_{i=2}^{M-1} \sigma_{21,i}^2 \|[{\bf V}_{21}]_i^H\hat{\bf v}_L \|^2].
\end{eqnarray}
Then, the upper bound in (\ref{prop4_1}) can be asymptotically upper bounded for large $M$ as
\begin{eqnarray}\label{RVQ_11}
\!\!IN_{21}^{up}\! \lesssim \!\alpha_{21}\!MP_1\!\biggl[ (1\!+\! 4\!\cdot\!2^{-\! \frac{B_{12}}{M}} ) \cos^2(\hat\phi_E\!-\!\theta_1) \!\!\!\nonumber\\+\! \sin^2(\hat\phi_E\!-\!\theta_1)\biggr].\!\!\!\!
\end{eqnarray}
Note that, as $\theta_1$ increases (close to MLB), $ \cos^2(\hat\phi_E-\theta_1)$ increases, which implies that the interference upper bound is more sensitive to the number of feedback bits. Here, we have also utilized that $E[\sigma_{21,M}^2] \leq  \int_{0}^{\infty}\lambda g_{{\bf H}_{11}^H{\bf H}_{11}}(\lambda)d\lambda =M $ in (\ref{RVQ_6}).
Similarly to (\ref{RVQ_7}), $S_{22}^{low}$ can be approximated as
\begin{eqnarray}\label{RVQ_12}
S_{22}^{low} \approx MP\left[ (4- 3 \cdot2^{- \frac{B_{22}}{M}})  \cos^2( \theta_2) +\sin^2(\theta_2)\right],
\end{eqnarray}
\begin{remark}\label{remark6}
From (\ref{RVQ_11}) and (\ref{RVQ_12}), the approximated lower bound of SINR can be written as
\begin{eqnarray}\label{RVQ_13}
\!\!\!\frac{S_{22}^{low}}{\!1\!+\! IN_{21}^{up\!}} \!\geq\! \frac{MP\left[ (4- 3 \cdot2^{- \frac{B_{22}}{M}})  \cos^2( \theta_2) +\sin^2(\theta_2)\right]}{\!1\!\!+\!\! \alpha_{21}\!M\!P_1\!\left[ \!(\!1\!+\! 4\!\cdot\!2^{- \!\frac{B_{12\!}}{M}}\! ) \!  \cos^2(\!\hat\phi_E\!\!-\!\!\theta_1\!) \!+\! \sin^2(\!\hat\phi_E\!\!-\!\!\theta_1\!)\!\right]\!}\!\!\!\nonumber\\ \triangleq SINR^{low}.\quad\quad\quad\quad
\end{eqnarray}
To find the optimal $B_{22}$ and $B_{12}$, by substituting $B_{12} = B - B_{22}$ into $SINR^{low}$ in (\ref{RVQ_13}), we can find $B_{22}$ maximizing $SINR^{low}$ numerically. That is,
\begin{eqnarray}\label{RVQ_13_1}
B_{22} =\arg\underset{B_{22}\in\{0,...,B\}}{\max} SINR^{low}.
\end{eqnarray}
Note that, if the target harvesting energy is small and the harvested energy from the information transmitter is enough ($P_1 =0$), $B_{22}$ maximizing (\ref{RVQ_13}) becomes equal to $B$. Similarly, when the path loss of the cross link is large enough ($\alpha_{21} \rightarrow 0$), $B_{22}$ maximizing (\ref{RVQ_13}) also becomes equal to $B$. That is,  we do not allocate feedback bits for the cross link. In contrast, when $P_1$ and $\alpha_{21}$ are large (the power of the interference signal becomes large), the SINR can be increased by allocating more bits to the cross link (i.e., by increasing $B_{12}$).
\end{remark}

\section{Simulation Results}
\label{sec:simulation}
Computer simulations have been performed to verify the proposed schemes. Throughout the simulations, we generate channel ${\bf H}_{ij}$
according to the i.i.d. zero-mean complex Gaussian distribution with a unit variance for $i = j$ and a variance $\alpha_{12} =\alpha_{21}= \alpha \in [0,1]$ (the relative path loss of the cross link compared to the direct link) as described in Section \ref{sec:systemmodel}. In addition, the path loss of the direct links is assumed to be $10^{-3/2}$ which implies that the path loss exponent is $3$ and $10m$ distance between Tx $i$ and Rx $i$ ($-30dB = 10\log_{10}10^{-3}$). The maximum transmit power is set as $P=50mW$ and the noise power is $1 \mu W$, unless otherwise stated.

\begin{figure}
\begin{center}
\begin{tabular}{c}
\includegraphics[height=4.5cm]{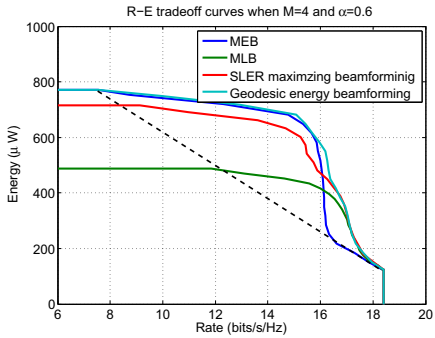}
\end{tabular}
\end{center}
\caption[Fig1_fullCSIT]
{ \label{Fig1_fullCSIT} Achievable R-E region when full CSIT is available at IAP with $M=4$ and $\alpha =0.6$. Full CSIT is required at EAP for SLER beamforming.}
\end{figure}

Fig. \ref{Fig1_fullCSIT} shows the achievable R-E region of four different energy beamforming schemes - MEB, MLB, SLER maximizing beamforming, and Geodesic beamforming, when $M= 4$, $\alpha= 0.6$, and full local CSIT is available at both EAP and IAP. Note that full CSIT at EAP is required for SLER maximizing beamforming. That is, the Algorithm 1 is utilized for SLER maximization, MEB, and MLB, while Algorithm 2 is exploited for Geodesic beamforming. Note that we can see that the R-E region of the Geodesic beamforming covers those of all other beamforming schemes, which is consistent with Theorem 1 and, as Remark \ref{remark2}, the SLER maximizing beamforming has a similar R-E region with the Geodesic beamforming. The dashed line indicates the R-E curves of the time-sharing of 1) the full-power rank-one MEB to EH MS at both EAP and IAP and 2) no transmission at EAP and waterfilling at IAP. Note that MEB shows worse performance than the time-sharing especially when the target required energy is small. That is, because the MEB causes large interference to the ID receiver, it is desirable that, for the low required harvested energy, the first transmitter turns off its power in the time slots where the second transmitter is assigned to exploit the waterfilling. Even in these slots, EH MS can harvest energy from IAP signal. In the remaining slots, EAP opts for a MEB with full power and IAP transfers its information to the ID receiver by steering its beam on EH receiver's channel ${\bf H}_{12}$. Accordingly, the transferred energy to EH MS will be maximized. In these slots, ID MS can also receive its information from IAP. 

\begin{figure}
\centering 
 \subfigure[]
  {\includegraphics[height=4.5cm]{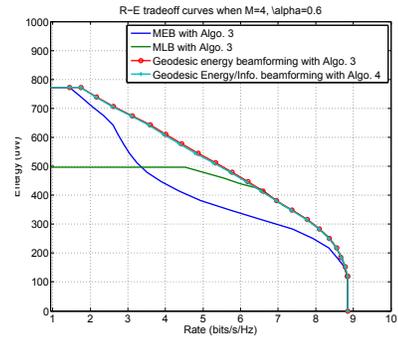}}\quad 
 \subfigure[]
  {\includegraphics[height=4.5cm]{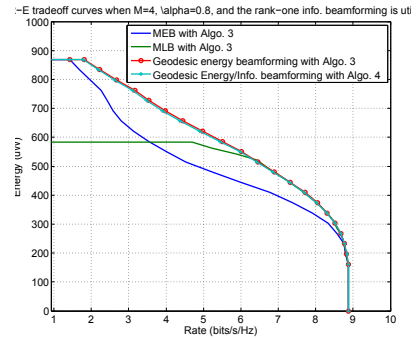}}
\caption[Fig2_partialCSIT]
{ \label{Fig2_partialCSIT} Achievable R-E region when Algorithm 3 and Algorithm 4 are exploited for $M=4$, (a) $\alpha =0.6$ and (b) $\alpha =0.8$.}
\end{figure}

Fig. \ref{Fig2_partialCSIT} shows the achievable R-E region when the rank-one information beamforming is utilized when $M=4$, (a) $\alpha =0.6$ and (b) $\alpha =0.8$. That is, when EAP exploits MEB, MLB, and Geodesic beamforming, ${\bf Q}_2$ is optimally determined by solving the optimization (P2) (Algorithm 3). In addition, the R-E region, when the Geodesic beamforming is exploited in both IAP and EAP (Algorithm 4), is also provided. Note that Algorithm 3 requires full CSIT at IAP, while Algorithm 4 requires partial CSIT at both IAP and EAP. We can see that the Geodesic beamforming in both IAP and EAP with Algorithm 4 exhibits the same performance with the optimal beamforming with Algorithm 3. We can see that the maximum harvesting energy with $\alpha = 0.8$ is higher than that with $\alpha = 0.6$ due to the larger harvested energy from IAP's signal. In addition, even though the overall achievable rates are smaller than those provided in Fig. \ref{Fig1_fullCSIT} due to the rank-one constraint at IAP, the maximum achievable harvesting energy is similar with that in Fig. \ref{Fig1_fullCSIT}. This is because the maximum achievable harvesting energy can be achieved when both IAP and EAP opt for the rank-one beamforming.

\begin{figure}
\begin{center}
\begin{tabular}{c}
\includegraphics[height=4.5cm]{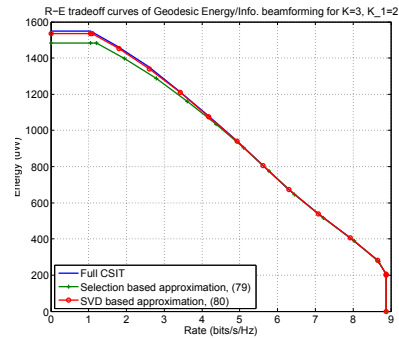}
\end{tabular}
\end{center}
\caption[Fig3_K_user]
{ \label{Fig3_K_user} Achievable R-E region of Geodesic energy/information beamforming for $M=4$, $K=3$, $K_1 =2$, and $\alpha = 0.5$.}
\end{figure}

Fig. \ref{Fig3_K_user} shows the achievable rate of the Geodesic information/energy beamforming for $M=4$, $\alpha = 0.5$, and $K=3$, where two EAP and one IAP coexist. The full local CSIT implies that the EMD is the largest singular vector of $[{\bf H}_{1j}^T,{\bf H}_{2 j}^T]^T$ as discussed in Section \ref{ssec:KuserIFC}. Note that the selection based method of (\ref{Discusseqn1}) exhibits worse performance than other schemes. Especially, because EMD vector is approximated in each transmitter, the maximum harvesting energy is smaller than those of other schemes. In contrast, the SVD based approximation of EMD as (\ref{Discusseqn2}) shows almost similar performance to the full CSIT.

\begin{figure}
\centering 
 \subfigure[]
  {\includegraphics[height=4.3cm]{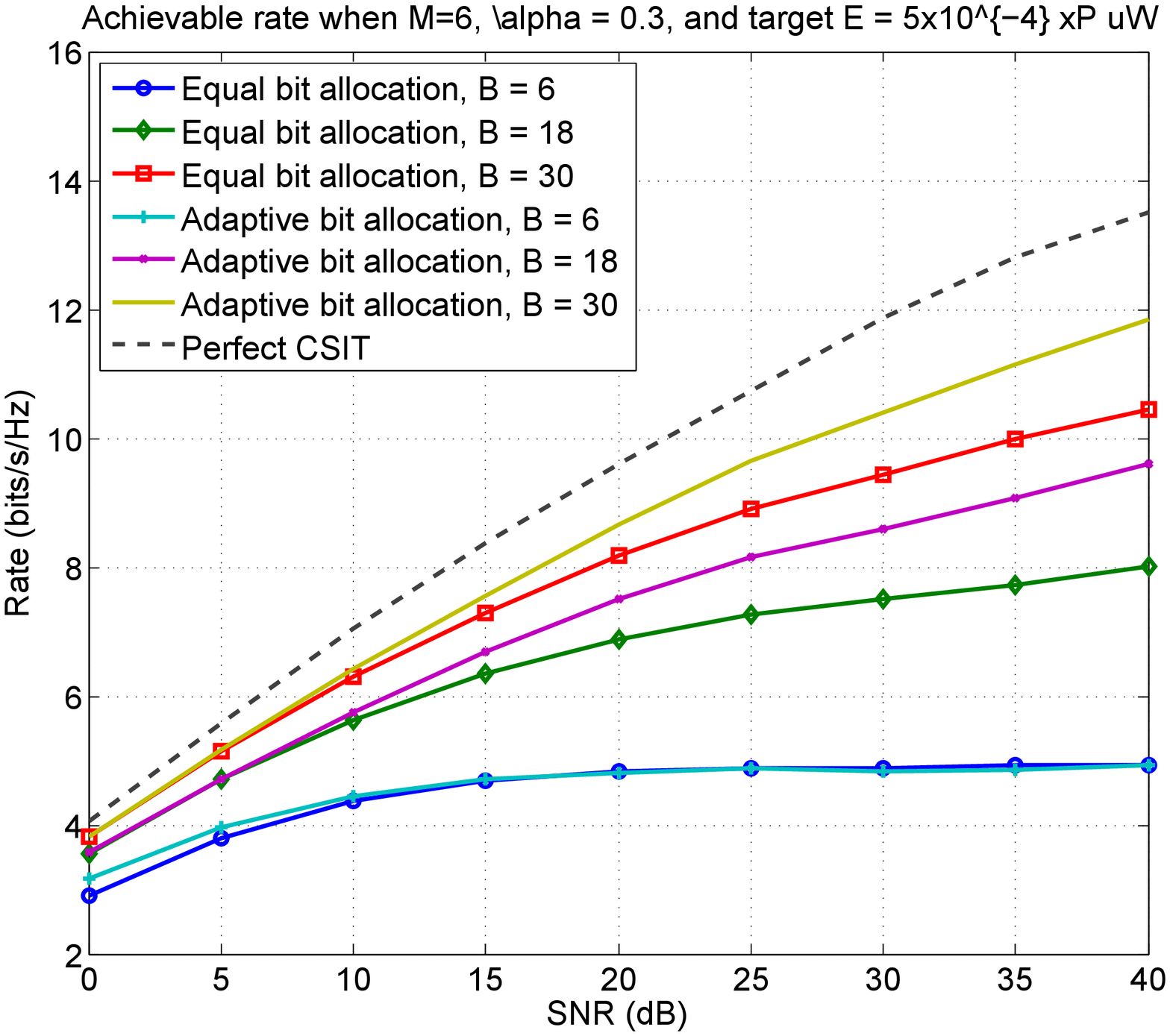}}\\
 \subfigure[]
  {\includegraphics[height=4.3cm]{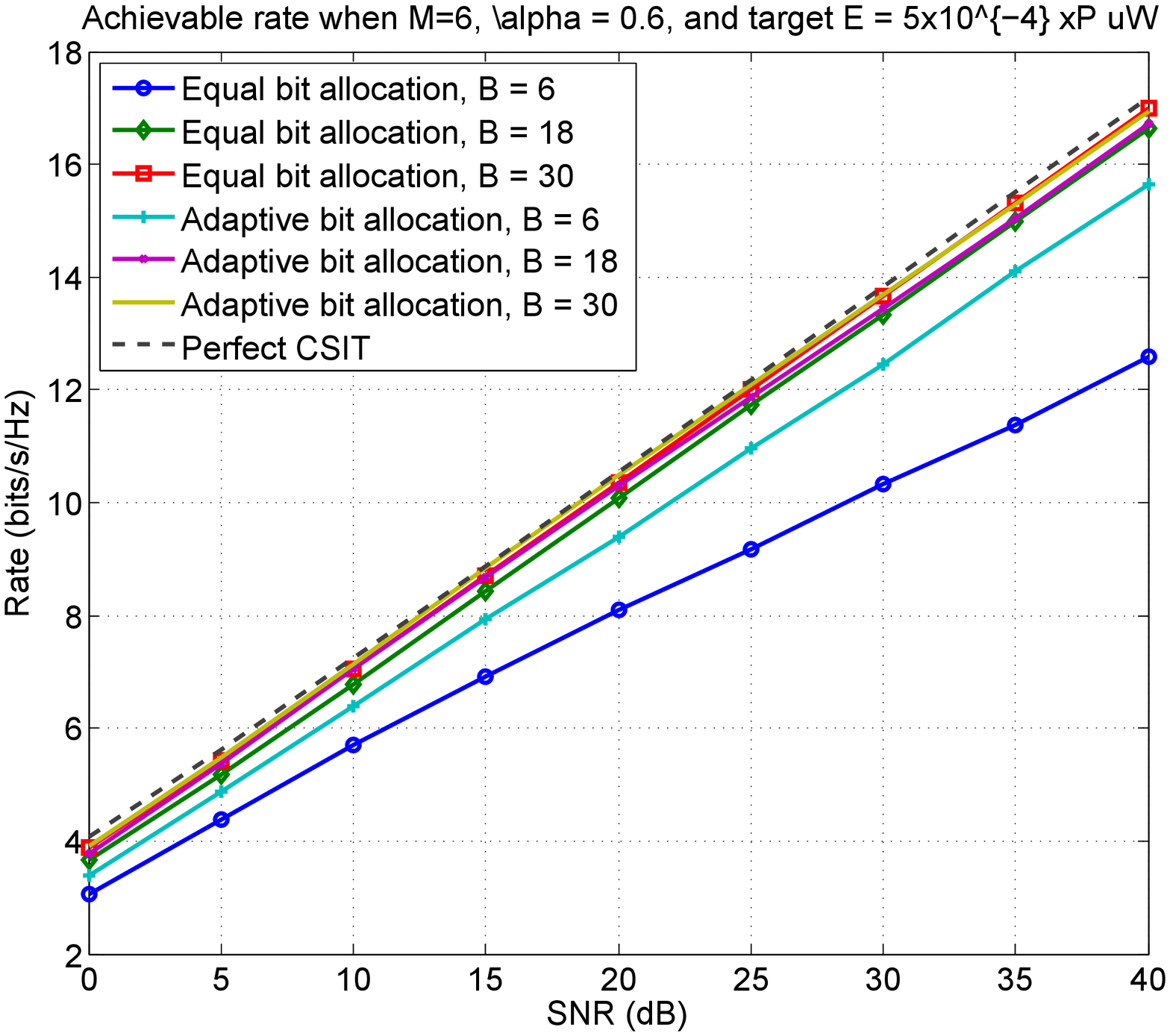}}\\
   \subfigure[]
  {\includegraphics[height=4.3cm]{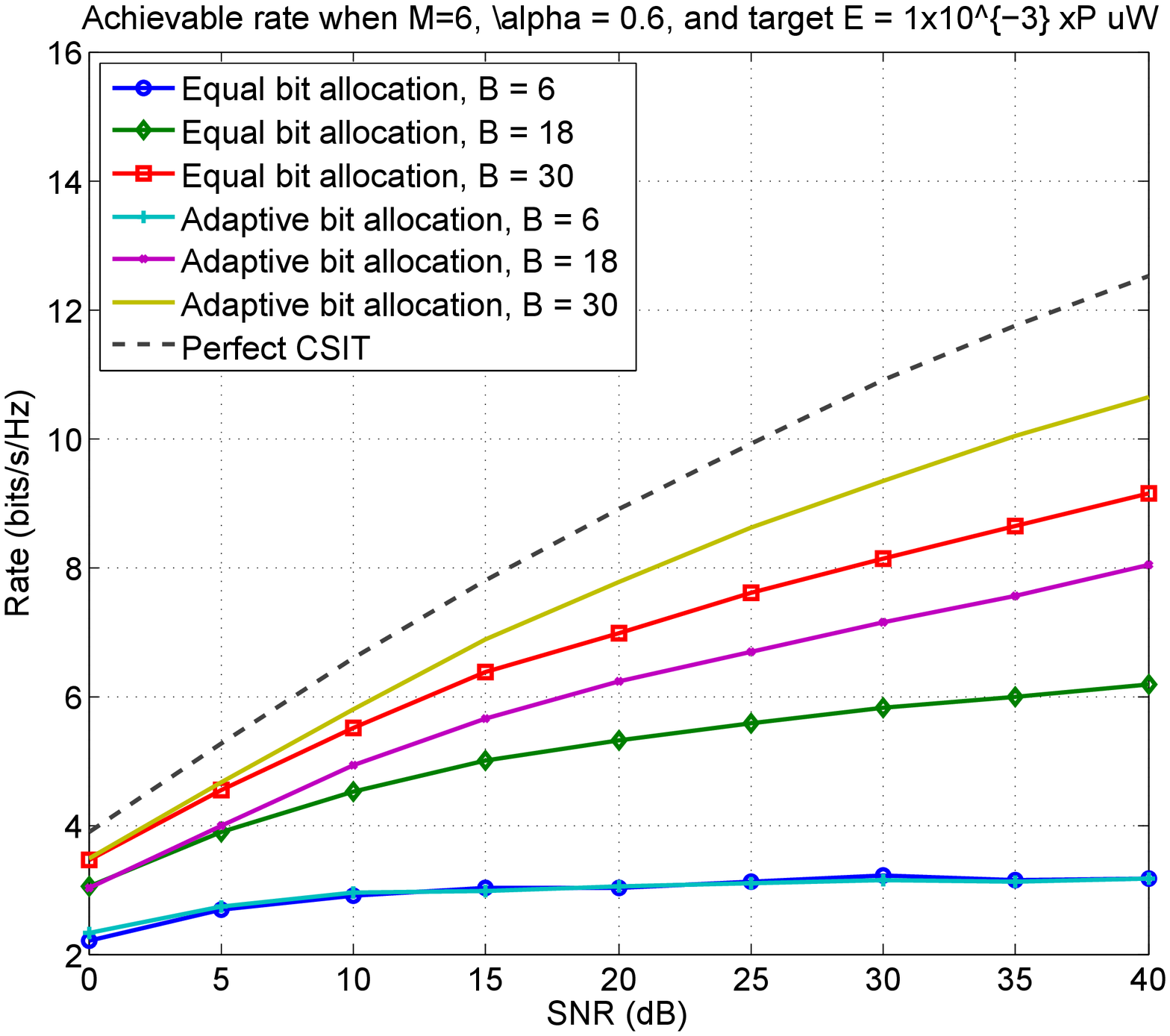}}
 \caption{Achievable rate when $M=6$, and (a) $\alpha=0.3$,  the target energy $\bar E = 5\times 10^{-4}P \mu W$, (b) $\alpha=0.6$ and $\bar E = 5\times 10^{-4}P \mu W$, and (c) $\alpha=0.6$ and $\bar E = 1\times 10^{-3}P \mu W$.} \label{Fig4_Imperfect}
\end{figure}

Fig. \ref{Fig4_Imperfect} shows the achievable rate of the Geodesic information/energy beamforming with limited feedback when $M=6$, and (a) $\alpha=0.3$ and the target energy $\bar E = 5\times 10^{-4}P \mu W$, (b) $\alpha=0.6$ and $\bar E = 5\times 10^{-4}P \mu W$, and (c) $\alpha=0.6$ and $\bar E = 1\times 10^{-3}P \mu W$. Here we have utilized the RVQ \cite{NRavindran} in quantizing ${\bf v}_E$, ${\bf v}_L$, ${\bf w}_I$, and ${\bf v}_L$. As the number of feedback bits ($B$) increases, the performances with limited feedback schemes become close to that of perfect CSIT. In addition, the adaptive bit allocation exhibits better performance than the equal bit allocation. Interestingly, as in Fig. \ref{Fig4_Imperfect}(a), when $\alpha$ is small (i.e., the cross-link path loss is large), the achievable rate saturates. This is because, when $\alpha$ is small, the interference from IAP cannot contribute to the harvested energy at EH MS effectively to satisfy the target energy. Therefore, IAP needs to steer its information beam to EH MS and EAP performs MEB. That is, the SINR at ID MS saturates as Tx power increases. In contrast, when $\alpha$ is large, EAP can steer its energy beam to the MLB and IAP can also steer its beam to ID MS. Therefore, the achievable rate increases proportionally with the SNR (or, transmit power).
Similar saturation can be found in Fig. \ref{Fig4_Imperfect}(c). That is, when the required target energy is large, then the rich interference environment is preferable to both ID/EH MSs to satisfy energy requirement and simultaneously maintain the information rate slope ({\it{degree of freedom}}).

\section{Conclusion}
\label{sec:conc}
In this paper, to reduce the feedback overhead of MSs for the JWIET in the two-user MIMO IFC, we have investigated a Geodesic energy beamforming scheme that requires partial CSI at the EAP. Furthermore, in the two-user MIMO IFC, we have proved that the Geodesic energy beamforming scheme is the optimal strategy. By adding a rank-one constraint on the transmit signal covariance of IAP, we can further reduce the feedback overhead to IAP by exploiting a Geodesic information beamforming scheme. Under the rank-one constraint of IAP's transmit signal, we prove that Geodesic information/energy beamforming approach is the optimal strategy for JWIET in the two-user MIMO. Furthermore, for the deployment of our proposed Geodesic information/energy beamforming schemes to the general K-user MIMO IFC, EAPs and IAPs should estimate the EMD and IMD with their partial CSIT, which can be done by the selection/SVD based approximations using the feedback information reported from MSs. By analyzing the achievable rate-energy performance statistically under the imperfect CSIT, we have proposed an adaptive bit allocation strategy for both EH MS and ID MS.

\useRomanappendicesfalse
\appendices

\section{Proofs of Proposition 1 and Proposition 2}\label{appndix1}
Note that $E_{11}(\theta_1) = \|{\bf H}_{11}{\bf v}_G(\theta_1)\|^2$. From (\ref{Partial4}), because ${\bf H}_{11}{\bf v}_E$ and ${\bf H}_{11}({\bf v}_E)^{\perp}{\bf u}_L$ are orthogonal to each other, we have
\begin{eqnarray}\label{Prop1_9}
\!\!&\!\!\|{\bf H}_{11}\!{\bf v}_G(\!\theta_1\!)\|^2 \!=\! \cos^2(\!\theta_1\!)\|{\bf H}_{11}{\bf v}_E\|^2\! +\!\sin^2(\!\theta_1\!) \| {\bf H}_{11}({\bf v}_E)^{\perp}\!{\bf u}_L\|^2 \!\!&\!\!\nonumber\\\!\!& \!\!= \cos^2(\!\theta_1\!)\sigma_{11,1}^2\!+\! \sin^2(\!\theta_1\!) \!\sum_{i=2}^{M}\!\alpha_i^2\sigma_{11,i}^2,\!\!&\!\!\!
\end{eqnarray}
where $\sum_{i=2}^{M-1}\alpha_i^2 = 1$ due to the fact that $\|({\bf v}_E)^{\perp}{\bf u}_L\|^2 = 1$. Accordingly, for $\theta'_1 >\theta''_1$,
\begin{eqnarray}\label{Prop1_10}
\!\!&\!\!\|{\bf H}_{11}\!{\bf v}_G(\theta'_1)\|^2\!-\!\|{\bf H}_{11}\!{\bf v}_G(\theta''_1)\|^2 \!=\! \sigma_{11,1}^2(\cos^2(\theta'_1) - \cos^2(\theta''_1))\nonumber\!\!&\!\!\\\!\!&\!\!+ \sum_{i=2}^{M}\alpha_i^2\sigma_{11,i}^2(\sin^2(\theta'_1) - \sin^2(\theta''_1)),\!\!&\nonumber\!\!\\\!\!&\!\!
= (\sum_{i=2}^{M}\alpha_i^2\sigma_{11,i}^2 - \sigma_{11,1}^2 )(\sin^2(\theta'_1) - \sin^2(\theta''_1)) <0\!\!&\!\!
\end{eqnarray}
which implies that $E_{11}(\theta'_1) < E_{11}(\theta''_1)$. Accordingly, $f(\theta_1)$ is monotonically decreasing with respect to ${\theta}_1$.

Similarly to what is done in the proof of Proposition \ref{prop3}, we can show that for $\theta'_1 >\theta''_1$, $IN_{21}(\theta'_1) < IN_{21}(\theta''_1)$ (Proposition \ref{prop3_IN}).


\section{Proof of Proposition 3}\label{appndix3}
Let $J({\bf Q}_1) \!\triangleq \!\log \det({\bf I}_{M} \!+\! {\bf H}_{22}^H({\bf I}_M \!+ \!{\bf H}_{21}{\bf Q}_1 {\bf H}_{21}^H)^{-1}{\bf H}_{22}{\bf Q}_2 ).
$ Then, from Lemma \ref{lem1}, the maximization of $J({\bf Q}_1)$ is equivalent with the minimization of $P_1 \|{\bf H}_{21}{\bf v}_G(\theta_1)\|^2 $, because
\begin{eqnarray}\label{Prop2_3}
\det({\bf I}_{M} + {\bf H}_{21}{\bf Q}_1 {\bf H}_{21}^H) = 1 + P_1 \|{\bf H}_{21}{\bf v}_G(\theta_1)\|^2.
\end{eqnarray}
Now, let us assume that, given the optimal $\theta_1^o$ and $P_1^o$, we have $\theta'_1$ such that $\eta(\theta'_1) > \eta(\theta_1^o)$. If $\|{\bf H}_{21}{\bf v}_G(\theta'_1)\|^2 \geq \|{\bf H}_{21}{\bf v}_G(\theta_1^o)\|^2$, then we can set $P_1' = P_1^o\frac{ \|{\bf H}_{21}{\bf v}_G(\theta_1^o)\|^2}{\|{\bf H}_{21}{\bf v}_G(\theta'_1)\|^2}$, resulting in $P_1'\|{\bf H}_{11}{\bf v}_G(\theta'_1)\|^2 = P_1^o\frac{ \|{\bf H}_{21}{\bf v}_G(\theta_1^o)\|^2}{\|{\bf H}_{21}{\bf v}_G(\theta'_1)\|^2}\|{\bf H}_{11}{\bf v}_G(\theta'_1)\|^2 \geq P_1^o\|{\bf H}_{11}{\bf v}_G(\theta_1^o)\|^2  $. That is, from Lemma \ref{lem1} (see also (\ref{Prop2_3})), $\theta'_1$ yields more harvested energy given the same achievable rate. For $\|{\bf H}_{21}{\bf v}_G(\theta'_1)\|^2 < \|{\bf H}_{21}{\bf v}_G(\theta_1^o)\|^2$, we have to consider two cases -- $\|{\bf H}_{11}{\bf v}_G(\theta'_1)\|^2 \geq \|{\bf H}_{11}{\bf v}_G(\theta_1^o)\|^2$ or $\|{\bf H}_{11}{\bf v}_G(\theta'_1)\|^2 < \|{\bf H}_{11}{\bf v}_G(\theta_1^o)\|^2$. The former case corresponds to the case that $\theta'_1$ with $P_1' = P_1^o$ yields more harvested energy and more achievable rate. For the latter case, from Propositions \ref{prop3} and \ref{prop3_IN}, $\theta'_1$ does not exist satisfying $\|{\bf H}_{21}{\bf v}_G(\theta'_1)\|^2 < \|{\bf H}_{21}{\bf v}_G(\theta_1^o)\|^2$ and $\|{\bf H}_{11}{\bf v}_G(\theta'_1)\|^2 < \|{\bf H}_{11}{\bf v}_G(\theta_1^o)\|^2$, simultaneously. Therefore, all cases contradict the statement that $\theta_1^o$ yields the boundary point of the achievable $C_{R-E}$ for the Geodesic energy beamforming.

\section{Proofs of Theorem 1 and Theorem 2}\label{appndix4}
\subsubsection{Proof of Theorem 1}
Let us assume that ${\bf v}_1^o$ is an optimal beamforming vector yielding a boundary point of the optimal R-E region. Then, from Proposition \ref{prop2}, the optimal solution implies that, there is no beamforming vector ${\bf v}_1$ that has
\begin{eqnarray}\label{thm1_1}
\frac{\|{\bf H}_{11}{\bf v}_1\|^2}{ \|{\bf H}_{21}{\bf v}_1\|^2 } > \frac{\|{\bf H}_{11}{\bf v}_1^o\|^2}{ \|{\bf H}_{21}{\bf v}_1^o\|^2}.
\end{eqnarray}
First, we define $\cos (\phi_{E1}^o )= |{\bf v}_E^H {\bf v}_1^o |, \quad \cos (\phi_{L1}^o) = |{\bf v}_L^H {\bf v}_1^o |$,
where $\phi_{E1}^o$ (resp. $\phi_{L1}^o$) is the principal angle between ${\bf v}_E$ (resp. ${\bf v}_L$) and ${\bf v}_1^o $. Then, similarly to (\ref{Prop1_9}), $\|{\bf H}_{11}{\bf v}_1^o\|^2$ and $\|{\bf H}_{21}{\bf v}_1^o\|^2$ can be represented as
\begin{eqnarray}\label{thm1_4}
\!\!\|{\bf H}_{11}{\bf v}_1^o\|^2\! =\! \cos^2 (\phi_{E1}^o) \sigma_{11,1}^2 \!+\! \sin^2(\phi_{E1}^o) \!\!\sum_{i=2}^M\alpha_{Ei}^2\sigma_{11,i}^2,\nonumber\!\!\\
\!\!\|{\bf H}_{21}{\bf v}_1^o\|^2\!=\! \cos^2 (\phi_{L1}^o) \sigma_{21,M}^2 \!+\! \sin^2(\phi_{L1}^o)\! \! \sum_{i=1}^{M\!-\!1}\alpha_{Li}^2\sigma_{21,i}^2,\!\!\!\!\!
\end{eqnarray}
where $ \sum_{i=2}^M\alpha_{Ei}^2 = 1$ and $ \sum_{i=1}^{M-1}\alpha_{Li}^2 = 1$. Note that for ${\bf v}'_1$ with $\phi'_{E1} <\phi_{E1}^o$, we have the inequality shown at the top of the next page.
\begin{figure*}[!t]
\scriptsize
\begin{eqnarray}\label{thm1_5}
\|{\bf H}_{11}{\bf v}'_1\|^2 - \|{\bf H}_{11}{\bf v}_1^o\|^2  &\!=\!&( \cos^2 (\phi'_{E1}) -  \cos^2 (\phi_{E1}^o) ) \sigma_{11,1}^2 + \sin^2(\phi'_{E1}) \sum_{i=2}^M\alpha_{Ei}'^2\sigma_{11,i}^2 - \sin^2(\phi_{E1}^o) \sum_{i=2}^M\alpha_{Ei}^2\sigma_{11,i}^2,\nonumber\\
&\!>\!& ( \cos^2 (\phi'_{E1}) -  \cos^2 (\phi_{E1}^o) ) \sigma_{11,1}^2 +( \sin^2 (\phi'_{E1}) -  \sin^2 (\phi_{E1}^o) ) \sigma_{11,2}^2\nonumber\\
&\!=\!& ( \cos^2 (\phi'_{E1}) -  \cos^2 (\phi_{E1}^o) ) (\sigma_{11,1}^2 - \sigma_{11,2}^2 ) > 0.
\end{eqnarray}
\hrulefill \vspace*{2pt}
\end{figure*}
This implies that as the principal angle between ${\bf v}_E$ and ${\bf v}_1$ (denoted as $\phi_{E1}$) decreases, $\|{\bf H}_{11}{\bf v}_1\|^2$ increases (monotonic decreasing). Similarly, we can find that $\|{\bf H}_{21}{\bf v}_1\|^2$ is monotonic increasing with respect to the principal angle $\phi_{L1}$ between ${\bf v}_L$ and ${\bf v}_1$. Therefore, to maximize
\begin{eqnarray}\label{thm1_5}
\!\!\frac{\|{\bf H}_{11}\!{\bf v}_1\|^2}{ \|{\bf H}_{21}\!{\bf v}_1\|^2 }\!=\! \frac{ \cos^2 (\phi_{E1}) \sigma_{11,1}^2 \!+ \!\sin^2(\phi_{E1}) \sum_{i=2}^M\alpha_{Ei}^2\sigma_{11,i}^2 }{ \cos^2 (\phi_{L1}) \sigma_{21,M}^2 \!+\! \sin^2(\phi_{L1}) \sum_{i=1}^{M-1}\alpha_{Li}^2\sigma_{21,i}^2 },\!\!
\end{eqnarray}
both $\phi_{E1}$ and $\phi_{L1}$ should be minimized.

Now, we assume that ${\bf v}_1^o$ is not on the Geodesic curve between $[{\bf V}_{11}]_1$ and $[{\bf V}_{21}]_M$. We can always find ${\bf v}_1$ such that $\phi_{E1}<\phi_{E1}^o$ and $\phi_{L1}<\phi_{L1}^o$ on the Geodesic curve. Note that the minimum value of $\phi_{E1}+\phi_{L1} = \phi_E$ as in (\ref{Partial6}) (see also Fig. \ref{geodesicOptimal}).
\subsubsection{Proof of Theorem 2}
Then, similarly to (\ref{Prop1_9}), $\|{\bf H}_{22}{\bf w}_2\|^2$ (resp. $\|{\bf H}_{12}{\bf w}_2\|^2$) can be represented as
\begin{eqnarray}\label{thm2_4}
\|{\bf H}_{22}{\bf w}_2\|^2 = \cos^2 (\phi_{I2}) \sigma_{22,1}^2 + \sin^2(\phi_{I2}) \sum_{i=2}^M\alpha_{Ii}^2\sigma_{22,i}^2,\nonumber\\
\|{\bf H}_{12}{\bf w}_2\|^2 = \cos^2 (\phi_{L2}) \sigma_{12,1}^2 + \sin^2(\phi_{L2}) \sum_{i=1}^{M-1}\alpha_{Li}^2\sigma_{12,i}^2,
\end{eqnarray}
where $\phi_{I2}^o$ (resp. $\phi_{L2}^o$) is the principal angle between ${\bf w}_I$ (resp. ${\bf w}_L$) and ${\bf w}_2$ and $ \sum_{i=2}^M\alpha_{Ii}^2 = 1$ and $ \sum_{i=2}^{M}\alpha_{Li}^2 = 1$. Note that, similarly to Theorem \ref{thm1}, we can find that $\|{\bf H}_{22}{\bf w}_2\|^2$ (resp. $\|{\bf H}_{12}{\bf w}_2\|^2$) is monotonic decreasing with respect to the principal angle $\phi_{I2}$ (resp. $\phi_{L2}$). From (\ref{PartialCSIT3}), to maximize the achievable rate and harvested energy, both $\phi_{I2}$ and $\phi_{L2}$ should be minimized. Then, we assume that ${\bf w}_2$ is not on the Geodesic curve between $[{\bf V}_{22}]_1$ and $[{\bf V}_{12}]_1$. We can always find ${\bf w}'_2$ such that $\phi'_{I2}<\phi_{I2}$ and $\phi'_{L2}<\phi_{L2}$ on the Geodesic curve. Note that the minimum value of $\phi_{I2}+\phi_{L2} = \phi_I$.


\bibliographystyle{IEEEtran}
\bibliography{IEEEabrv,myref}

\end{document}